%% file: main.tex
\documentclass{article}
\usepackage{tikz}
\usepackage{tikz-cd}
\usepackage{amsmath,amsthm,amssymb}
\usepackage{mathpazo}
\usepackage{stmaryrd}
\usepackage{subcaption}
\usepackage{enumerate}
\usepackage{scalerel}
\usepackage{graphicx,float}
\usepackage{multicol}
\usepackage{xcolor}
\usepackage{bussproofs}
\usepackage{etoc}
\usepackage[disable]{todonotes}
\usepackage{hyperref}
\hypersetup{
    colorlinks,
    linkcolor={red!50!black},
    citecolor={blue!50!black},
    urlcolor={blue!80!black}
}
\usepackage{natbib}
\newenvironment{bprooftree}
{\leavevmode\hbox\bgroup}
{\DisplayProof\egroup}

\newtheorem{theorem}{Theorem}[section]

\newtheorem{proposition}[theorem]{Proposition}
\newtheorem{lemma}[theorem]{Lemma}

\theoremstyle{definition}
\newtheorem{definition}[theorem]{Definition}
\newtheorem{example}[theorem]{Example}

\usetikzlibrary{shapes}
\usetikzlibrary{positioning}
\usetikzlibrary{calc}
\usetikzlibrary{decorations.pathmorphing}

\tikzset{curve/.style={settings={#1},to path={(\tikztostart)
    .. controls ($(\tikztostart)!\pv{pos}!(\tikztotarget)!\pv{height}!270:(\tikztotarget)$)
    and ($(\tikztostart)!1-\pv{pos}!(\tikztotarget)!\pv{height}!270:(\tikztotarget)$)
    .. (\tikztotarget)\tikztonodes}},
    settings/.code={\tikzset{quiver/.cd,#1}
        \def\pv##1{\pgfkeysvalueof{/tikz/quiver/##1}}},
    quiver/.cd,pos/.initial=0.35,height/.initial=0}

\tikzset{tail reversed/.code={\pgfsetarrowsstart{tikzcd to}}}
\tikzset{2tail/.code={\pgfsetarrowsstart{Implies[reversed]}}}
\tikzset{2tail reversed/.code={\pgfsetarrowsstart{Implies}}}
\tikzset{no body/.style={/tikz/dash pattern=on 0 off 1mm}}

\input{macros.tex}

\title{Beyond Nonexpansive Operations in Quantitative Algebraic Reasoning}
\author{Matteo Mio \and Ralph Sarkis \and Valeria Vingudelli}

\begin{document}

\maketitle
\begin{abstract}
   
    The framework of quantitative equational logic has been successfully applied to reason about algebras whose carriers are metric spaces and operations are nonexpansive. We extend this framework in two orthogonal directions: algebras endowed with generalised metric space structures, and operations being nonexpansive up to a lifting. We apply our results to the algebraic axiomatisation of the \L ukaszyk--Karmowski distance on probability distributions, which has recently found application in the field of representation learning on Markov processes.
    \end{abstract}
\section{Introduction}\label{intro}
\input{intro.tex}

\section{Background}\label{background}
\input{background.tex}

\section{Quantitative Reasoning with Liftings}\label{section_definitions}
\input{definitions.tex}

\section{Term Monad and Free Quantitative Algebras}\label{section_termfree}
\input{results.tex}

\section{Examples}\label{examples}
\input{examples.tex}
\section{Conclusion}\label{conclusion}
\input{conclusion.tex}
\bibliographystyle{plainnat}
\bibliography{biblio}
\newpage
\section{Appendix}\label{appendix}
\input{appendix.tex}

\end{document}

%% file: macros.tex
\newrobustcmd{\R}{\mathbb{R}}
\newrobustcmd{\N}{\mathbb{N}}
\newrobustcmd{\Q}{\mathbb{Q}}
\newrobustcmd{\F}{\mathbb{F}}
\newrobustcmd{\Z}{\mathbb{Z}}
\newrobustcmd{\mC}{\mathcal{C}}
\newrobustcmd{\mD}{\mathcal{D}}
\newrobustcmd{\mDmet}{\widehat{\mathcal{D}}}
\newrobustcmd{\mP}{\mathcal{P}}
\newrobustcmd{\one}{\mathbb{1}}

\newrobustcmd{\bigplus}{\mathop{\mbox{\Large$+$}}}
\newrobustcmd{\norm}[1]{\left\lVert #1 \right\rVert}
\newrobustcmd{\inp}[2]{\left\langle #1, #2 \right\rangle}
\newrobustcmd{\id}{\mathrm{id}}
\newrobustcmd{\gln}{\text{GL}_n}
\newrobustcmd{\opp}[1]{#1^{\mathrm{op}}}
\newrobustcmd{\supp}[1]{\mathrm{supp}\left( #1 \right)}
\newrobustcmd{\inl}{\mathsf{inl}}
\newrobustcmd{\inr}{\mathsf{inr}}
\newrobustcmd{\act}[1]{\stackrel{#1}{\mapsto}}
\newrobustcmd{\unact}[1]{\stackrel{#1}{\mapsfrom}}
\allowdisplaybreaks
\newrobustcmd{\point}{\star}
\newrobustcmd{\bpoint}{\ast}
\newrobustcmd{\jeq}[1]{\stackrel{\eqref{#1}}{=}}
\newrobustcmd{\haus}{H}
\newrobustcmd{\kant}{K}
\newrobustcmd{\hk}{\haus\kant}
\newrobustcmd{\distr}{\varphi}
\newrobustcmd{\ddistr}{\Phi}
\newrobustcmd{\distrb}{\psi}
\newrobustcmd{\ddistrb}{\Psi}
\newrobustcmd{\distrc}{\theta}
\newrobustcmd{\EM}{\mathbf{EM}}
\newrobustcmd{\KL}{\mathrm{Kl}}
\newrobustcmd{\EMs}{\mathrm{EM}_{\mathrm{s}}}
\newrobustcmd{\LK}[1]{{#1}_{\textnormal{\L K}}}
\newrobustcmd{\LKp}[2]{L_{#2}(#1)}
\newrobustcmd{\dirac}[1]{\delta_{#1}}
\newrobustcmd{\LKAlg}{\mathbf{Alg}_{\textnormal{\L K}}}
\newrobustcmd{\mmin}{\mathfrak{min}}
\newrobustcmd{\bn}{\mathbf{n}}
\newrobustcmd{\suc}{\mathsf{succ}}
\newrobustcmd{\Set}{\mathbf{Set}}
\newrobustcmd{\FRel}{\mathbf{FRel}}
\newrobustcmd{\PreMet}{\mathbf{PreMet}}
\newrobustcmd{\PQMet}{\mathbf{PQMet}}
\newrobustcmd{\PSMet}{\mathbf{PSMet}}
\newrobustcmd{\DMet}{\mathbf{DMet}}
\newrobustcmd{\PMet}{\mathbf{PMet}}
\newrobustcmd{\QMet}{\mathbf{QMet}}
\newrobustcmd{\MMet}{\mathbf{MMet}}
\newrobustcmd{\SMet}{\mathbf{SMet}}
\newrobustcmd{\GMet}{\mathbf{GMet}}
\newrobustcmd{\Met}{\mathbf{Met}}
\newrobustcmd{\UMet}{\mathbf{UMet}}
\newrobustcmd{\op}{\mathsf{op}}
\newrobustcmd{\sem}[1]{\llbracket #1 \rrbracket}
\newrobustcmd{\term}[1]{T_{#1}}
\newrobustcmd{\eqns}[1]{\mathcal{V}_{#1}}
\newrobustcmd{\eqn}{\phi}
\newrobustcmd{\eqnb}{\psi}
\newrobustcmd{\clauses}[1]{\mathcal{H}_{#1}}
\newrobustcmd{\qterm}[1]{\widehat{T}_{#1}}
\newrobustcmd{\Mod}{\mathbf{Mod}}
\newrobustcmd{\Alg}{\mathbf{Alg}}
\newrobustcmd{\forget}{U}
\newrobustcmd{\algA}{\mathbb{A}}
\newrobustcmd{\algB}{\mathbb{B}}
\newrobustcmd{\Cat}{\mathbf{C}}
\newrobustcmd{\Catb}{\mathbf{D}}
\newrobustcmd{\mon}{M}
\newrobustcmd{\monb}{T}
\newrobustcmd{\vars}{X}
\newrobustcmd{\sig}{\Sigma}
\newrobustcmd{\qsig}{\widehat{\Sigma}}
\newrobustcmd{\satisfies}{\vDash}
\newrobustcmd{\satp}[1]{\vDash^{#1}}
\newrobustcmd{\assign}{\iota}
\newrobustcmd{\subst}{\sigma}
\newrobustcmd{\termeq}[1]{\equiv_{#1}}
\newrobustcmd{\termd}[1]{d_{#1}}
\newrobustcmd{\relgen}[1]{\vdash_{#1}}
\newrobustcmd{\qeqclass}[1]{\langle #1 \rangle}
\newrobustcmd{\collapse}[1]{\mathfrak{q}_{#1}}
\newrobustcmd{\LKD}{\mathcal{D}^{\textnormal{\L K}}}

%% file: intro.tex

Equational reasoning and algebraic methods are widespread in all areas of computer science, and in particular in program semantics. Indeed, initial algebra semantics and monads are cornerstones of the modern theory of functional programming and allow us to reason about inductive definitions, computational effects and specifications in a formal way (see, e.g., \cite{DBLP:journals/iandc/Moggi91,DBLP:conf/rex/RuttenT93, Hyland2006}). In elementary terms, this is due to the fact that many objects of interest in programming are free algebras of some algebraic theory, i.e., a \emph{signature} $\sig$ together with a set of equational \emph{axioms} $E$ between $\sig$--terms. Examples include:  finite sets (free algebras of the theory of semilattices)
$$
 \Sigma  = \{ \vee : 2 \} \ \ \ \ \  E = \begin{Bmatrix} x\vee y = y\vee x,\  x \vee x = x,\\
  x \vee(y \vee z) = (x \vee y) \vee z\end{Bmatrix} 
$$
\noindent
finite lists (free monoids), finitely supported distributions (free convex algebras) \emph{etc}. Since free algebras are (up to isomorphism) term algebras---i.e., sets of $\sig$--terms modulo the congruence relation $\termeq{E}$ generated from the axioms $E$ using the deduction rules of the syntactic apparatus of \emph{equational logic}---they are easy to manipulate formally in a computer. 

Objects definable as free algebras, as in the framework outlined above, are sets $X$ equipped with operations of type $X^n \rightarrow X$.
This means it is not straightforward, or even possible, to describe objects that are sets endowed with some additional structure such as, e.g., a metric $d:X^2\rightarrow[0,1]$. To address this limitation, in a series of recent papers (including \cite{DBLP:conf/lics/BacciMPP18,radu2016, DBLP:conf/lics/MardarePP17, DBLP:conf/calco/BacciMPP21,DBLP:conf/lics/MardarePP21}), the authors have proposed the notion of \emph{quantitative algebras}: algebras whose carriers are metric spaces. 

At the syntactic level, the apparatus of equational logic is replaced by a deductive system allowing the derivation of judgments of the form $s=_\varepsilon t$, where $s,t$ are $\sig$--terms and $\varepsilon \in [0,1]$, with the intended meaning that $d(s,t)\leq \varepsilon$. These judgments are derived using quantitative inferences, i.e.,  deduction rules of the form: 
\begin{center}
$\{s_1 =_{\varepsilon_1} t_1, \dots , s_n=_{\varepsilon_n} t_n\} \vdash s=_\varepsilon t$.
\end{center}
In particular, the deductive system includes rules such as: 
\begin{center}
$\emptyset \vdash x =_0 x\ \ \ \ \ \ \ \ \ \ \ \ \ \ \ \ \ \
\{x =_\varepsilon y\} \vdash y =_\varepsilon x $
\end{center}
\begin{center}
$\{ x =_{\varepsilon_1} y, y=_{\varepsilon_2}z\}  \vdash x=_{\varepsilon_1 + \varepsilon_2} z$
\end{center}
corresponding to properties of metrics such as reflexivity ($d(x,x)=0$), symmetry ($d(x,y)=d(y,x)$) and triangular inequality ($d(x,y)\leq d(x,y)+d(y,z)$). 
A \emph{quantitative theory} over a signature $\sig$ is generated from a set of quantitative inferences, playing the role of implicational axioms,
 by closing under deducibility in the apparatus. 
Models of quantitative theories are  \emph{quantitative algebras}, \todo{MM: Made small modification after skype discussion}
which are metric spaces $(A,d)$ equipped with interpretations $\sem\op: A^n \to A$ of the operations such that for each $\op \in \Sigma$
$$d(\sem{\op}(a_1,...,a_n), \sem{\op}(a'_1, ... , a'_n)) \leq \max\{d(a_i,a'_i)\}_{1\leq i \leq n}.$$
This is equivalent to requiring that
$\sem{\op}:(A^n, d_\times) \rightarrow (A,d)$ is nonexpansive (also known as $1$--Lipschitz), with  $d_\times$ being the (categorical) product metric on $A^n$.
This is reflected in the deductive system by a rule called \textsf{NE}:
\[\{x_i=_{\varepsilon_i} y_i\}_{1\leq i \leq n}
\vdash \op(x_1,...,x_n) =_{{\max(\varepsilon_1,...,\varepsilon_n)}} \op(y_1, ... , y_n).\]

Consider, for example, 
the theory of \emph{quantitative semilattices} of \cite{radu2016} having signature  $\sig \!= \!\{ \vee  : 2\}$ and implicational axioms (we just write $s=_\varepsilon t$ for $\emptyset\vdash s=_\varepsilon t$):
\[x \vee y =_0 y \vee x \quad  x \vee x =_0 x \quad
x \vee (y \vee z) =_0 (x\vee y) \vee z \]
These just state the usual axioms of semilattices. Indeed, since in any metric space it holds that $d(x,y)=0$ implies $x=y$, the judgment $s=_0 t$ expresses equality.  
From these axioms, further quantitative inferences can be obtained using the deductive apparatus, like the \textsf{NE} rule:\todo{I slightly changed because I didn't like "derive further implicational axioms"}
\[\left\{ x=_{\varepsilon_1} x^\prime, y=_{\varepsilon_2} y^\prime \right\} \vdash x \vee y =_{\max(\varepsilon_1,\varepsilon_2)} x^\prime \vee y^\prime\]
which expresses that the interpretation of the binary operation $\vee: 2$ must be nonexpansive.

Given a quantitative theory over a signature $\Sigma$ generated by a set of implicational axioms $E$, we have a category $\Alg(\Sigma, E)$ consisting of quantitative algebras modelling the theory and their homomorphisms, i.e., nonexpansive maps $f:(A,d_A)\rightarrow (B,d_B)$ preserving all operations $\sem{\op}$.
Among the main results of  \cite{radu2016, DBLP:conf/lics/MardarePP17, DBLP:conf/lics/BacciMPP18} the following is of key importance:\todo{Rewrote this part}

\vspace{2mm}
%
%
\noindent
\textbf{Theorem 3.3 in  \cite{DBLP:conf/lics/BacciMPP18}}. 
The free quantitative algebra generated by a metric space $(A,d)$ exists in $\Alg(\Sigma, E)$ and is isomorphic to the quantitative term algebra $\term{\Sigma,E}(A,d)$.\vspace{2mm}


More can be said if the implicational axioms $E$ have a constrained form, where all the terms in their premises are variables:
$
x_1 =_{\varepsilon_1} y_1, \dots, x_n =_{\varepsilon_n} y_n \vdash s=_\varepsilon t
$.
In this case, which covers several interesting examples (e.g., quantitative semilattices), we have a stronger result:
\vspace{2mm}

\noindent
\textbf{Theorem 4.2 in  \cite{DBLP:conf/lics/BacciMPP18}}. The Eilenberg--Moore category $\EM(\term{\Sigma,E})$ of the term monad $\term{\Sigma,E}$ is isomorphic to the category
$\Alg(\Sigma, E)$.

\vspace{2mm}

Several interesting metric spaces can be identified with free quantitative algebras. For example the collection of non-empty finite subsets of $(A,d)$, endowed with the Hausdorff metric and interpreting  $\sem{\vee} = \cup$ (union), can be shown (see \cite{radu2016}) to be isomorphic to the free quantitative semilattice generated by the metric space $(A,d)$. 

%


\subsection{Beyond Metric Spaces and Nonexpansive Maps}

The main purpose of this paper is to extend the framework of \cite{DBLP:conf/lics/BacciMPP18} outlined above, while maintaining its key characteristics and properties, in order to reason equationally about additional interesting mathematical objects which do not fit the constraints of the original framework.

We immediately discuss a specific example arising from recent research in the field of learning and artificial intelligence \cite{DBLP:journals/corr/abs-2106-08229}, which will serve as a main motivation. 
Other examples are discussed in Section \ref{examples}.
  In \cite{DBLP:journals/corr/abs-2106-08229}, the authors have developed new techniques for representation learning on Markov processes based on the \L ukaszyk–Karmowski (\L K for short) distance \cite{lukaszyk04}. This is a distance  $\LK{d}:\mD X\times\mD X\rightarrow [0,1]$ on finitely supported distributions on a set $X$ endowed with an arbitrary map $d:X^2\rightarrow[0,1]$ (i.e., $(X,d)$ is not necessarily a metric space). Even if $d$ is a metric, the \L K distance $\LK{d}$ does not satisfy all axioms of metric spaces. Specifically the reflexivity property is in general not satisfied: $\LK{d}(\distr,\distr) \neq 0$. However,  $\LK{d}$ always satisfies the symmetry and triangular inequality axioms (see Equations \eqref{eq:symm} and \eqref{eq:trineq} in Section \ref{background}) and, therefore, $(\mD X, \LK{d})$ is a \emph{diffuse metric space}  (see \cite{DBLP:journals/corr/abs-2106-08229} or Section \ref{sec:gmets} for precise definitions).
  %
If we consider the convex algebra operation $+_p: \mD(X)\times \mD(X)\rightarrow \mD(X)$ on probability distributions defined by
\[(\distr +_p \distrb)(x) = p\distr(x) + (1-p)\distrb(x),\]
then it can be shown (see Lemma \ref{lem:plusp_not_ne}) that $+_p$ fails to be nonexpansive with respect to the \L K distance: 
$$
  \LK{d} ( \distr +_p \distr^\prime,  \distrb +_p \distrb^\prime  )   > \max\{ \LK{d}(\distr, \distrb) , \LK{d}(\distr^\prime, \distrb^\prime)\}.
$$

Thus we have an interesting mathematical object,  the diffuse metric space $(\mD(X), \LK{d})$, whose underlying set $\mD(X)$ is the free convex algebra over the set $X$ (see, e.g., \cite{jacobs:2010}), not fitting the framework of \cite{DBLP:conf/lics/BacciMPP18} due to two reasons: (1) $\LK{d}$ is not a metric, and (2) the algebraic (convex algebra) operation $+_p$ is not nonexpansive. 

Our contribution is to extend the framework of \cite{radu2016, DBLP:conf/lics/MardarePP17, DBLP:conf/lics/BacciMPP18} along two orthogonal axes in order to accomodate examples (see Section \ref{examples}) such as the one just discussed.

\paragraph{First extension axis:} our framework can be instantiated on structures $(X,d)$ where $d:X^2\rightarrow [0,1]$ is a \emph{generalised metric} such as any of the following (see Section \ref{background} for details): an ultrametric, metric, pseudometric, quasimetric, diffuse metric or just a fuzzy relation (i.e., $d$ unconstrained).

This first contribution is natural, yet requires some technical care. 
Most notably, we need to carefully distinguish in the deductive apparatus between the notions of equality $(=)$ and zero distance $(=_0)$. This is due to the fact that, unlike the case of metric spaces, in generalised metric spaces (e.g., pseudometric or diffuse metrics) it does not hold that $d(x,y)=0$ implies $x=y$. 
As a consequence, the identification of $=$ and $=_0$ is generally unsound. Our deductive apparatus, unlike that of   \cite{DBLP:conf/lics/BacciMPP18}, will therefore handle both ordinary equations $(s=t)$ and quantitative equations ($s=_\varepsilon t$), connected by the following congruence principle:
\[ x = y  \Rightarrow \left( (x =_\varepsilon z \Rightarrow y =_\varepsilon z ) \textnormal{ and }
( z =_\varepsilon x \Rightarrow z =_\varepsilon y ) \right).\]

\paragraph{Second extension axis:} our framework can deal with quantitative algebras whose operations are not nonexpansive with respect to the categorical product. The motivating example being the diffuse metric space $(\mD(X), \LK{d})$ with the convex combination operation $+_p$ discussed earlier. This is in our opinion the main conceptual and technical contribution of the paper. 



To achieve this flexibility, we consider \emph{lifted signatures} $\qsig=\{ \op_i \!:\! n_i\! :\! L_{\op_i}\}_{i\in I}$. Each operation $\op$ has an arity $n\in \N$, as for standard signatures, and is further equipped with 
a lifting which maps any generalised metric space $(X,d)$ to a generalised metric space 
$(X^{n}, L_\op(d))$ whose underlying set is the product set $X^{n}$, subject to some technical constraints.

In this new setting, quantitative algebras for a lifted signature $\qsig$ are (generalised) metric spaces $(X,d)$ in $\GMet$  where, for each $\op\in \sig$, the interpretation $\sem{\op}: X^n\rightarrow X$ is \emph{nonexpansive up to $L_\op$}, namely:
\[\sem{\op}: (X^n, L_\op(d))\rightarrow (X,d) \quad \textnormal{ is nonexpansive}.\]
At the syntactic level, our deductive apparatus replaces the \textsf{NE} rule of \cite{DBLP:conf/lics/BacciMPP18}
%
with a rule denoted by \textsf{$L$--NE} (see Definition \ref{defn:deduce}) expressing that each  $\op: n:L_\op \in \qsig$ is nonexpansive up to $L_\op$.

The framework of \cite{DBLP:conf/lics/BacciMPP18} can be seen as a particular case of ours by taking $\GMet = \Met$ and restricting all $L_{\op_i}$ to be the standard $n$--ary (categorical) product in $\Met$:
$L_{\op_i}(X,d) = (X^{n_i}, d_{\times})$.

%
%

\subsection{Outline and Main Results}
After presenting some background material in Section \ref{background}, we introduce in Section \ref{section_definitions} our new framework for quantitative reasoning based on liftings, and we prove the soundness of the associated deductive apparatus. In Section \ref{section_termfree}, we define the term monad and
we recover the key results of the framework of \cite{DBLP:conf/lics/BacciMPP18} in our new ``lifted'' setting. In particular we obtain proofs of the corresponding variants of Theorem 3.3 (free algebras exist and are term algebras) and Theorem 4.2 ($\EM(\qterm{\qsig,E}) \cong \Alg(\qsig,E)$) from \cite{DBLP:conf/lics/BacciMPP18}.
We give examples of applications of our new apparatus in Section \ref{examples}, covering in particular the interesting case of the \L K diffuse metric on probability distributions.
Full proofs can be found in the appendix.
%
%
%

%% file: background.tex
\subsection{Monads}
We present some definitions and results regarding monads. We assume the reader is familiar with basic concepts of category theory (see, e.g., \cite{Awodey}). Facts easily derivable from known results in the literature are systematically marked as ``Proposition'' throughout the paper.

\begin{definition}[Monad]\label{monad:main_definition}
A \emph{monad} on a category $\Cat$ is a triple $(\mon, \eta, \mu)$ comprising a functor $\mon\colon\Cat \rightarrow \Cat$ together with two natural transformations: a \emph{unit} $\eta\colon \id_{\Cat} \Rightarrow \mon$, where $\id_{\Cat}$ is the identity functor on $\Cat$, and a \emph{multiplication} $\mu \colon \mon^{2} \Rightarrow \mon$, satisfying
$\mu \circ \eta\mon = \mu \circ \mon\eta = \id_{\mon} $ and $\mu\circ \mon\mu = \mu \circ\mu\mon$.
\end{definition}
A monad $\mon$ has an associated category of $M$--algebras.
\begin{definition}[$M$--algebras]\label{def:algebra-of-a-monad}
Let $(\mon,\eta,\mu)$ be a monad on $\Cat$. An \emph{algebra} for $\mon$ (or \emph{$\mon$--algebra}) is a pair $(A,\alpha)$ where $A\in\Cat$ is an object and $\alpha:\mon (A)\rightarrow A$ is a morphism such that (1) $ \alpha \circ  \eta_A = \id_A$ and (2) $\alpha\circ \mon \alpha= \alpha \circ \mu_A $ hold. An \emph{$\mon$--algebra morphism} between two $\mon$--algebras $(A,\alpha)$ and $(A^\prime,\alpha^\prime)$ is a morphism $f:A\rightarrow A^\prime$ in $\Cat$ such that
$f\circ \alpha = \alpha^\prime \circ \mon(f)$. The category of $\mon$--algebras and their morphisms, denoted by $\EM(\mon)$, is called the Eilenberg--Moore category for $\mon$.
\end{definition}

%
%

\subsection{Universal Algebra}
\label{subsection_universal}
We recall basic definitions and results from universal algebra, \cite{univalgebrabook} is a standard reference.
\begin{definition}[Signature]
    A \emph{signature} is a set $\sig$ containing operations symbols each with an arity $n \in \N$. We denote $\op:n \in \sig$ for a symbol $\op$ with arity $n$ in $\sig$. With some abuse of notation, we also denote with $\sig$ the functor $\sig: \Set \rightarrow \Set$ with the following action:
    \[\sig(A) := \coprod_{\op:n\in \sig} A^n \quad \sig(f):= \coprod_{\op:n \in \sig} f^n.\] 
\end{definition}
\begin{definition}[$\sig$--algebra]
    A \emph{$\sig$--algebra} is an algebra for the functor $\sig$. Equivalently, it is a set $A$ equipped with a set $\sem{\sig}_A$ 
    of interpretations of the operation symbols, i.e., for every $\op: n\in \sig$ there is a function $\sem{\op}_A: A^n \rightarrow A$ in $\sem{\sig}_A$. We call $A$ the \emph{carrier set}. A \emph{homomorphism} between two $\sig$--algebras with carrier sets $A$ and $B$ is a function $f: A \rightarrow B$ preserving $\sem{-}$, i.e.,     satisfying $\forall \op :n\in \sig, \forall a_1, \dots, a_n$,
    \[f(\sem{\op}_A(a_1,\dots,a_n)) = \sem{\op}_B(f(a_1),\dots,f(a_n)).\]
    The category of $\sig$--algebras and their homomorphisms is denoted $\Alg(\sig)$.
\end{definition}
\begin{definition}[Term algebra]
    Let $\sig$ be a signature and $A$ be a set. We denote with  $\term{\sig}A$ the set of terms built from $A$ using the operations in $\sig$, i.e., the set inductively defined as follows: $a \in \term{\sig}A$ for any $a \in A$, and 
    $\op(t_1,\dots, t_n) \in \term{\sig}A$  for any  $ \op:n \in \sig$ and $t_1,\dots t_n \in \term{\sig}A$. The set $\term{\sig}A$ has a canonical $\sig$--algebra structure with the interpretation of the operations $\op:n\in\sig$, defined as:
    \[\sem{\op}(t_1,\dots, t_n) = \op(t_1,\dots,t_n).\]
    It is called the \emph{term algebra over $A$} and denoted $\term{\sig}A$ (like its carrier set). We often identify elements $a\in A$ with the corresponding terms $a\in  \term{\sig}A$.
\end{definition}

\begin{definition}[Term monad]
The assignment $A \mapsto \term{\sig}A$ can be turned into a functor $\term{\sig}:\Set\rightarrow \Set$ by inductively defining, 
for any function $f: A \rightarrow B$, the homomorphism $\term{\sig}f:\term{\sig}A \rightarrow \term{\sig}B$ as follows: for any $a \in A$, $(\term{\sig}f)(a) = f(a)$, and $\forall \op:n \in \sig$ and $\forall t_1,\dots t_n \in \term{\sig}A$, 
    \[\term{\sig}f(\op(t_1,\dots, t_n)) = \op(\term{\sig}f(t_1),\dots, \term{\sig}f(t_n)).\]
   This becomes a monad by defining the unit  $\eta^{\sig}_A: A \rightarrow \term{\sig}A$ as mapping $a\in A$ to the term $a\in \term{\sig}A$, and  the multiplication  $\mu^{\sig}_A : \term{\sig}(\term{\sig}A)\rightarrow \term{\sig}A$ as mapping a term built out of terms $t(t_1,\dots, t_n)$ 
   to the flattened term $t(t_1,\dots, t_n)$. 
    We call $(\term{\sig}, \eta^{\sig}, \mu^{\sig})$ the \emph{term monad} for $\sig$.
\end{definition}
\begin{proposition}\label{prop:algtermisalgsig}
    For any signature $\sig$, $\Alg(\sig) \cong \EM(\term{\sig})$.
\end{proposition}


For the rest of this paper, let $\vars$ be a fixed countable set of variables. An \emph{interpretation} of $\vars$ in  a $\sig$--algebra $\algA = (A,\sem{\sig})$ is a map $\assign: \vars \rightarrow A$. The interpretation extends to arbitrary $\term{\sig}\vars$ terms by inductively defining $\sem{-}^{\assign}: \term{\sig}\vars \rightarrow A$ as:
\[\sem{x}^{\assign} = \assign(x) \text{ and }\sem{\op(t_1,\dots,t_n)}^{\assign} = \sem{\op}\left( \sem{t_1}^{\assign},\dots,\sem{t_n}^{\assign} \right).\]
In cases where $\assign: \vars \rightarrow \term{\sig}A$ is an interpretation in a term algebra, we denote $\sem{-}^{\assign}$ with $\assign^*$ to emphasize that its action is straightforward. It can be seen as a completely syntactical rewriting procedure, as $\assign^*$ takes a term in $\term{\sig}\vars$ and replaces all occurrences of $x$ with the term $\assign(x)$. 
\begin{definition}[Equations and their models]
    An \emph{equation} over $\sig$ is a pair of $\sig$--terms over $\vars$, i.e., an element of $\term{\sig}\vars \times \term{\sig}\vars$ which we denote $s=t$. 
We say a $\sig$--algebra $\algA = (A,\sem{\sig})$ \emph{satisfies} an equation $s=t$, denoted $\algA \satisfies s=t$, if for any $\assign: \vars \rightarrow A$, $\sem{s}^{\assign} = \sem{t}^{\assign}$. We write $\algA \satp{\assign} s= t$ when the equality holds for a particular interpretation $\assign$. Given a set $E$ of equations over $\sig$, we denote by $\Alg(\sig,E)$ the full subcategory of $\Alg(\sig)$ of all algebras that satisfy all equations in $E$.
    \end{definition}
    
  \begin{definition}
  A congruence relation on $\algA =(A,\sem{\sig}_A)\in \Alg(\sig)$  is an equivalence relation $R\subseteq A^2$ such that for every $\op:n\in \sig$, if $(a_1,b_1)\in R$, \dots, $(a_n,b_n)\in R$ then it holds that $(\sem\op_A(a_1,\dots, a_n) ,  \sem\op_A(b_1,\dots, b_n))\in R$. If $R$ is a congruence then the interpretation of each $\op\in\sig$ is well-defined on the set $A/R$ of $R$--equivalence classes, by:
\[\sem\op_{A/R} ([a_1]_R,\dots, [a_n]_R)= [\sem\op_A(a_1,\dots, a_n))]_R\]
Then we have the algebra $\algA /R=(A/R, \sem{\sig}_{A/R})$.
  \end{definition}
  


\begin{definition}[Term monad, with equations] 
    Let $\sig$ be a signature, $E$ a set of equations over $\sig$, and $A$ a set. Denote with $\termeq{E_A}$ the smallest congruence on the term algebra $\term{\sig}A$ such that $(\term{\sig}A)/{\termeq{E_A}}   \in \Alg(\sig,E)$, i.e., 
    $(\term{\sig}A)/{\termeq{E_A}}$ satisfies all equations in $E$. We define a variant of the term monad denoted $\term{\sig,E}$ that sends a set $A$ to $\term{\sig}A/{\termeq{E_A}}$. Given a function $f: A \rightarrow B$, we define the function $\term{\sig,E}f:\term{\sig,E}A \rightarrow \term{\sig,E}B$ using the already defined $\term{\sig}f$: for any $t \in \term{\sig}A$, $\term{\sig,E}f([t]_{\termeq{E_A}}) = \left[ \term{\sig}f(t) \right]_{\termeq{E_B}}$. One can check that $\term{\sig,E}f$ is well-defined and makes $\term{\sig,E}$ into a functor. In fact, it is a monad with unit  $\eta^{\sig,E}_A = a \mapsto [a]_{\termeq{E_A}}$ and multiplication $$ \mu^{\sig,E}_A  = \left[ t([t_1]_{\termeq{E_A}},\dots, [t_n]_{\termeq{E_A}}) \right]_{\termeq{E_{\term{\sig,E}A}}} \mapsto [t(t_1,\dots, t_n)]_{\termeq{E_A}}.$$
  We call $(\term{\sig,E}, \eta^{\sig,E}, \mu^{\sig,E})$ the \emph{term monad} for $(\sig,E)$.
\end{definition}

\begin{proposition}\label{prop:algtermisalgsigeq}
    For any signature $\sig$ and any set $E$ of equations over $\sig$, $\Alg(\sig,E) \cong \EM(\term{\sig,E})$.
\end{proposition}
A corollary of the above proposition is that the free $(\sig,E)$--algebra over a set $A$ is $(\term{\sig}A/{\termeq{E_A}}, \sem{\sig})$, with the canonical interpretation of operations:
$$\sem{\op}([t_1]_{\termeq{E_A}},\dots, [t_n]_{\termeq{E_A}}) = [\op(t_1,\dots,t_n)]_{\termeq{E_A}}.$$

\subsection{Generalized Metric Spaces}\label{sec:gmets}

\begin{definition}[$\FRel$]
    A \emph{fuzzy relation} on a set $A$ is a map $d: A\times A \rightarrow [0,1]$. A morphism between two fuzzy relations $(A,d)$ and $(B,\Delta)$ is a map $f: A \rightarrow B$ that is \emph{nonexpansive} (also referred to as $1$--Lipschitz\todo{Added Lip here}) namely, $\forall a,a' \in A, \Delta(f(a),f(a'))\leq d(a,a')$. We denote by $\FRel$ the category of fuzzy relations and nonexpansive maps.
\end{definition}
Here is a non-exhaustive\footnote{A wider class of implicational constraints can be handled in our framework. We chose these five as running example since they are well known.} list of constraints on fuzzy relations that have been considered in the literature.
\begin{align}
    \forall a,b \in A,\quad & d(a,b) = d(b,a) &&\label{eq:symm}\\
    \forall a\in A,\quad  & d(a,a) = 0 &&\label{eq:refl}\\
    \forall a,b \in A,\quad  & d(a,b) = 0 \implies a=b&&\label{eq:idofind}\\
    \forall a,b,c \in A,\quad  & d(a,c)\leq d(a,b)+d(b,c) &&\label{eq:trineq}\\
    \forall a,b,c \in A,\quad  & d(a,c) \leq \max\{d(a,b),d(b,c)\} &&\label{eq:strongtrineq}
\end{align}
Each has a somewhat standard name, \eqref{eq:symm} is symmetry, \eqref{eq:refl} is indiscernibility of identicals or reflexivity, \eqref{eq:idofind} is identity of indiscernibles, \eqref{eq:trineq} is triangle inequality, and \eqref{eq:strongtrineq} is strong triangle inequality. Restricting $\FRel$ to relations that satisfy a subset of the axioms above, we get many categories of interest whose objects were studied at least once in the literature. 
\begin{equation*}
{\fontsize{8.4}{8.4}
\begin{tikzcd}
	&[-14pt]& \DMet &[-5pt]& \MMet \\
	&&&& \PMet \\
	\FRel && \PSMet &&&&[-10pt] \Met & \UMet & {} \\
	&&&& \SMet \\
	&& \PQMet && \QMet
	\arrow["{\text{\eqref{eq:symm}, \eqref{eq:trineq}}}"{description}, no head, from=3-1, to=1-3]
	\arrow["{\text{\eqref{eq:refl}, \eqref{eq:trineq}}}"{description}, no head, from=3-1, to=5-3]
	\arrow["{\text{\eqref{eq:symm}, \eqref{eq:refl}}}"{description}, no head, from=3-1, to=3-3]
	\arrow["{\text{\eqref{eq:trineq}}}"{description}, no head, from=3-3, to=2-5]
	\arrow["{\text{\eqref{eq:idofind}}}"{description}, no head, from=1-3, to=1-5]
	\arrow["{\text{\eqref{eq:refl}}}"{description}, no head, from=1-3, to=2-5]
	\arrow["{\text{\eqref{eq:idofind}}}"{description, pos=0.8}, no head, from=3-3, to=4-5]
	\arrow["{\text{\eqref{eq:idofind}}}"{description}, no head, from=5-3, to=5-5]
	\arrow["{\text{\eqref{eq:symm}}}"{description, pos=0.8}, no head, from=5-3, to=2-5]
	\arrow["{\text{\eqref{eq:refl}}}"{description}, no head, from=1-5, to=3-7]
	\arrow["{\text{\eqref{eq:idofind}}}"{description}, no head, from=2-5, to=3-7]
	\arrow["{\text{\eqref{eq:trineq}}}"{description}, no head, from=4-5, to=3-7]
	\arrow["{\text{\eqref{eq:symm}}}"{description}, no head, from=5-5, to=3-7]
	\arrow["{\text{\eqref{eq:strongtrineq}}}"{description}, no head, from=3-7, to=3-8]
\end{tikzcd}}
\end{equation*}
For example, \emph{metrics} ($\Met$) are fuzzy relations that satisfy axioms \eqref{eq:symm}--\eqref{eq:trineq}, \emph{pseudometrics} ($\PMet$) satisfy \eqref{eq:symm}, \eqref{eq:refl} and \eqref{eq:trineq}, and \emph{diffuse metrics} ($\DMet$) satisfy \eqref{eq:symm} and \eqref{eq:trineq}. Other examples include: quasimetrics ($\QMet$), pseudoquasimetrics ($\PQMet$), metametrics ($\MMet$), semimetrics ($\SMet$), pseudosemimetrics ($\PSMet$), ultrametrics ($\UMet$).  Different notions of morphisms between these objects have been considered (e.g.: continuous functions, contracting maps, etc.) but, for our purposes, we will work with full subcategories of $\FRel$ and hence keep nonexpansiveness as the only condition on morphisms. This choice implies that isomorphisms of fuzzy relations are bijections that preserve distances. \todo{Added comment on distance pres.} In the sequel, we write $\GMet$ for a category of \emph{generalized metric spaces}, which can stand for any full subcategory of $\FRel$ satisfying a fixed subset of axioms \eqref{eq:symm}--\eqref{eq:strongtrineq}.


All products and coproducts exist in $\GMet$ and are easy to define. Let $\{(A_i, d_i)\mid i \in I\}$ be a non-empty family of generalized metric spaces. The product is $(\prod_{i \in I} A_i, \sup_{i \in I} d_i)$, 
with $\sup_{i \in I} d_i: \left( \prod_{i \in I} A_i \right)\times \left( \prod_{i \in I} A_i \right) \rightarrow [0,1]$ defined for $\vec{a},\vec{b} \in \prod_{i \in I} A_i$ as:
\[(\sup_{i \in I} d_i)(\vec{a},\vec{b}) = \sup_{i \in I}d_i(\vec{a}_i,\vec{b}_i).\]
We \todo{Added} denote the $\sup$--metric $\sup_{i \in I} d_i$ just as $d_\times$ when the index set $I$ is clear.
The coproduct 
is given by $\amalg_{i \in I} d_i: \left( \coprod_{i \in I} A_i \right)\times \left( \coprod_{i \in I} A_i \right) \rightarrow [0,1]$, defined for $a \in A_j$ and $b \in A_k$ as:
\[(\amalg_{i \in I} d_i)(a,b) = \begin{cases}
    d_j(a,b) &\text{if } j=k\\1&\text{otherwise}
\end{cases}\]
The empty product, i.e., the terminal object, is given by $d_{\mathbf{1}} : \{\ast\} \times \{\ast\} \rightarrow [0,1]$, defined by \[d_{\mathbf{1}}(\ast,\ast) = \begin{cases}
    0&\text{if constraint \eqref{eq:refl} holds in $\GMet$}\\1&\text{otherwise}
\end{cases}.\]
The empty coproduct, i.e., the initial object, is the only possible fuzzy relation on the empty set (which vacuously satisfies all the axioms that must hold in $\GMet$).

\begin{definition}[Isometric embedding]
    A nonexpansive map $f:(A,d) \rightarrow (B,\Delta)$ is an \emph{isometry} if for any $a,a' \in A$, $\Delta(f(a),f(a')) = d(a,a')$. An \emph{isometric embedding} is an isometry that is injective.\footnote{In the category $\Met$ of metric spaces, any isometry is injective, but this is not true for all $\GMet$.} For any generalized metric space $(A,d)$ and subset $A' \subseteq A$, the inclusion $i: (A',d|_{A'}) \rightarrow (A,d)$ is an isometric embedding.
\end{definition}

%% file: definitions.tex

We introduce in this section our novel framework. In Subsection \ref{subsection_qa} we present the notion of \emph{liftings of signatures}, the associated concept of quantitative $\qsig$--algebras and define the classes $\Alg(\qsig,S)$ definable by  sets of implicational axioms. In Subsection \ref{section:apparatus} we define the syntactical deductive apparatus 
used to reason about equality and distance in quantitative $\qsig$--algebras, and prove the soundness theorem, stating that the syntactic apparatus guarantees correct derivations.


\subsection{Generalized Quantitative Algebras}\label{subsection_qa}
In what follows, a given category $\GMet$ is fixed. 
\begin{definition}
Given functors $F:\Set\rightarrow\Set$ and $L: \GMet \rightarrow \GMet$ we say that $L$ is a \emph{lifting} of $F$ (from $\Set$ to $\GMet$) if the following diagram commute, where $U$ is the expected forgetful functor:
\begin{equation*}
\begin{tikzcd}[row sep=1.9em, column sep=scriptsize]
	\GMet & \GMet \\
	\Set & \Set
	\arrow["U"', from=1-1, to=2-1]
	\arrow["F"', from=2-1, to=2-2]
	\arrow["L", from=1-1, to=1-2]
	\arrow["U", from=1-2, to=2-2]
\end{tikzcd}
\end{equation*} 
\end{definition}
Hence, for any lifting $L$, on objects we have $L(A,d) = (F(A), d')$ for some $d'$ which we denote with $d' = L(d)$. We will interchangeably use both notations $L(A,d)$ and $(F(A), L(d))$.
\begin{definition}
A lifting $L$  \emph{preserves isometric embeddings} if, whenever $f:(A,d)\rightarrow (B,\Delta)$ is an isometric embedding then $L(f):L(A,d) \rightarrow L(B,\Delta)$ is also an isometric embedding. 
\end{definition}
 Informally, this property holds when $L$ is compatible with the operation of taking subspaces. In the rest of this paper, we will be only interested in liftings that preserve  isometric embeddings and often just refer to them as liftings. 

\begin{example}\label{example_1}
Take as $\GMet$ the category $\Met$ of metric spaces. Consider the functor $F=\id$. Then, for any lifting $L$, $(A,d) \stackrel{L}{\mapsto} (A, L(d))$, so $L(d)$ is a distance on $A$. As examples of liftings of $F$ preserving isometric embeddings, we list: 
\begin{enumerate}
\item the identity:  $L(d)(a,a^\prime) = d(a,a^\prime)$, 
\item the scaling: $L(d)(a,a^\prime) = r\cdot d(a,a^\prime)$ for $r\in (0,1)$, 
\item the discrete distance: $L(d)(a,a^\prime) = 1$ if  $a\neq a^\prime$.
\end{enumerate}
Similarly, consider $F=(-)^2$, i.e., $F(A) = A\times A$ and $F(f) = f\times f$. In this case, $L(d)$ is a distance on $A\times A$.
Examples of liftings preserving isometric embeddings include:
\begin{itemize}
\item the standard product distance: $L(d)( ( a_1, a_1^\prime),  ( a_2 ,a_2^\prime) ) = \max\{d(a_1,a_2), d(a_1^\prime, a_2^\prime)\}$,
\item the discrete distance: $L(d)( ( a_1, a_1^\prime),  ( a_2 ,a_2^\prime) ) = 1$ if $ ( a_1, a_1^\prime) \neq ( a_2 ,a_2^\prime)$.
\end{itemize} 
\end{example}

We note, as in some of the examples above, that for any $\GMet$, if $F$ is the $n$--ary product endofunctor $(-)^n$ on $\Set$, then the $n$--ary product $d_\times$ in $\GMet$ is a lifting of $F$ preserving isometric embeddings. We refer to it as the $\sup$--product lifting and denote it by $L_\times$. Accordingly, for $n=0$, the lifting $L_\times$ maps any object to the terminal object in $\GMet$ and, for $n=1$,  the lifting $L_\times$ is the identity functor. 
 \todo{MM: Edited.}


\begin{definition}[Nonexpansiveness up to lifting]
Let $F:\Set\rightarrow \Set$ and $L$ a lifting of $F$. Let $(A,d)$ and $(B,\Delta$) in $\GMet$. We say that a function $f:F(A)\rightarrow B$ is \emph{nonexpansive up to $L$} (or \emph{$L$--nonexpansive}) if 
$f: (F(A),L(d)) \rightarrow (B,\Delta)$ is nonexpansive (i.e., it is a morphism in $\GMet$).
\end{definition}


\begin{example}\label{example_2}
As in the previous example, fix $\GMet=\Met$ and consider $F=\id$. Consider the metric space $([0,1],d)$, the unit interval with its standard Euclidean metric (i.e. $d(x,y) = |x-y|$), and the map $f:F([0,1])\rightarrow [0,1]$ defined as $f(x) = x^2$. If we take as lifting of $F$ the identity lifting $L$ from Example \ref{example_1} (i.e., the lifting $L_\times$) the function $f$ is not $L$--nonexpansive because, e.g., $\frac{4}{10} = d(\frac{6}{10}, 1) < d((\frac{6}{10})^2, 1^2) = \frac{64}{100}$. By contrast, if we take as lifting $L$ the discrete lifting then $f$ is trivially $L$--nonexpansive. In fact, any function $f:[0,1]\rightarrow [0,1]$ is nonexpansive up to the discrete lifting.
\end{example}

We are now ready to introduce the concept of \emph{lifted signature}, which extends the usual notion of signature $\Sigma$ from universal algebra.

\begin{definition}[Lifted signature]
    Given a signature $\sig=\{ \op_i : n_i\}_{i\in I}$, a \emph{lifting}  of $\sig$ to $\GMet$ is a choice, for each $i\in I$, of lifting $L_{\op_i}$ of the $n_i$--ary product $(-)^{n_i}: \Set \rightarrow \Set$.
%
An operation symbol $\op$ with arity $n$ and associated lifting $L_\op$ is now denoted $\op: n:L_\op$.
    We denote lifted signatures $\qsig = \{ \op_i: n_i :L_{\op_i}\}_{i \in I}$ to clearly distinguish them from ordinary signatures.
\end{definition}

Note that, given any signature $\Sigma$, it is possible to obtain a lifted signature $\qsig$ by choosing, for each $\op:n\in\Sigma$, the $\sup$--product lifting $L_\times$ of $(-)^n$. 

As in the classical case, any lifted signature $\qsig$ gives rise to an endofunctor on $\GMet$ (denoted $\qsig$ too) with the following action:
\[\qsig(A,d) := \hspace{-10pt}\coprod_{\op:n:L_\op\in \qsig} L_\op(A,d) \quad \quad \qsig(f) := \hspace{-10pt}\coprod_{\op:n:L_\op\in \qsig}L_\op(f).\]
\begin{definition}[Quantitative $\qsig$--algebra]
    A \emph{quantitative $\qsig$--algebra} is an algebra for the functor $\qsig$. Equivalently, it is a generalised metric space $(A,d) \in \GMet$ equipped with a set $\sem{\qsig}_A$ 
    of interpretations of operation symbols, as follows: every $\op:n:L_\op \in \qsig$ is interpreted as a map $\sem{\op}_A : A^{n} \rightarrow A$ which is $L_\op$--nonexpansive, i.e., such that  
    $$ \sem{\op}_A: (A^n,L_\op(d)) \rightarrow (A,d) \textnormal{ is nonexpansive.}$$
We call $(A,d)$ the \emph{carrier space}. A \emph{homomorphism} between two quantitative $\qsig$--algebras with carrier spaces $(A,d)$ and $(B,\Delta)$ is a nonexpansive map $f:A \rightarrow B$ preserving all operations, i.e., $\forall \op:n : L_\op \in\qsig$ and $\forall a_1,\dots, a_n \in A$,  \[f(\sem{\op}_A(a_1,\dots,a_n)) = \sem{\op}_B(f(a_1),\dots,f(a_n)).\]
    The category of quantitative $\qsig$--algebras is denoted $\Alg(\qsig)$.
\end{definition}

We remark that, in the particular case of $\GMet=\Met$ and $\qsig$ being the $\sup$--product lifting of some signature $\Sigma$, the notion of \emph{quantitative $\qsig$--algebra} coincides with that of quantitative algebra for the signature $\Sigma$ of the framework of \cite{DBLP:conf/lics/BacciMPP18}. 

Any quantitative $\qsig$--algebra yields a $\sig$--algebra by applying the forgetful functor to $\forget:\GMet\rightarrow \Set$ because $\forget\sem{\op}_A$ has type $A^n \rightarrow A$ and morphisms in $\Alg(\qsig)$ are already $\sig$--algebra homomorphisms. We obtain the following commutative square of forgetful functors.

\begin{equation}\label{diag:forgetsquare}
    \begin{tikzcd}[row sep=scriptsize, column sep=scriptsize]
        {\Alg(\qsig)} & \GMet \\
        {\Alg(\sig)} & {\mathbf{Set}}
        \arrow[from=1-1, to=2-1]
        \arrow[from=2-1, to=2-2]
        \arrow[from=1-1, to=1-2]
        \arrow[from=1-2, to=2-2]
    \end{tikzcd}
\end{equation}
\begin{definition}[Equations]
    Given a quantitative $\qsig$--algebra $\algA := (A,d,\sem{\qsig})$ and an equation $e \in \term{\sig} X \times \term{\sig} X$, we say that $\algA$ \emph{satisfies} $e$, denoted $\algA \satisfies e$, if its underlying $\sig$--algebra satisfies $e$.
\end{definition}
\begin{definition}[Quantitative equation]
    A \emph{quantitative equation} in the signature $\qsig$ is an element $e \in \term{\sig} X \times \term{\sig} X \times [0,1]$, i.e. a triple comprising two $\sig$--terms $s$ and $t$ and a real number $\varepsilon \in [0,1]$. We denote it $s=_\varepsilon t$. We say that $\algA := (A,d,\sem{\qsig})$ \emph{satisfies} $s=_\varepsilon t$, denoted $\algA \satisfies s=_\varepsilon t$, if for any variable assignment $\assign: X \rightarrow A$, $d(\sem{s}^\assign,\sem{t}^\assign)\leq \varepsilon$. We write $A \satp{\assign} s=_{\varepsilon} t$ when the inequality holds for a particual assignment $\assign$.
\end{definition}
Let $\eqns{\sig}X = \term{\sig}X\times \term{\sig}X \cup \term{\sig}X\times \term{\sig}X\times [0,1]$ denote the set of equations and quantitative equations over the signature $\sig$ and variables $X$. We use the letter $\eqn$ to range over $\eqns{\sig}X$.

Following \cite{DBLP:conf/lics/BacciMPP18}, we will consider classes of $\qsig$--algebras axiomatised by  (quantitative) equational implications, rather than just (quantitative) equations. While this level of generality is not required in many applications, as several useful examples (see Section \ref{examples}) 
are purely (quantitative) equational, it allows for a direct comparison of our results and those of \cite{DBLP:conf/lics/BacciMPP18}.
%
\begin{definition}[Horn clauses]
\label{def:hornclauses}
    In the sequel, we denote $\clauses{\sig}(X) = \mP(\eqns{\sig}X) \times \eqns{\sig}X$ the set of (possibly infinitary) \emph{Horn clauses} over the signature $\sig$ and variables $X$. A Horn clause $H \in \clauses{\sig}(X)$ is denoted $\bigwedge_{i \in I} \eqn_i \Rightarrow \eqn$ as its intended semantics is that $\eqn$ holds whenever each $\eqn_i$ holds. More formally, we say that an algebra $\algA = (A,d,\sem{\qsig}_A) \in \Alg(\qsig)$ \emph{satisfies} a clause $H= \bigwedge_{i \in I} \eqn_i \Rightarrow \eqn$, denoted $\algA \satisfies H$, if for any variable assignment $\assign: X \rightarrow A$, $\algA \satp{\assign} \eqn$ whenever $\algA \satp{\assign} \eqn_i$ for every $i$. We write $\algA \satp{\assign} H$ when the implication is true for a particular assignment $\assign$. We call $H = \bigwedge_{i \in I} \eqn_i \Rightarrow \eqn$ \emph{basic} if each premise $\eqn_i$ is a (quantitative) equation between variables: $\eqn_i$ is either of the form $x = y$ or $x =_\varepsilon y$, for $x,y\in X$ and $\varepsilon \in [0,1]$.

    Given a set $S \subseteq \clauses{\sig}(X)$, we denote $\Alg(\qsig,S)$ the full subcategory of $\Alg(\qsig)$ containing all quantitative $\qsig$--algebras that satisfy all clauses in $S$.
\end{definition}
\subsection{Syntactic Apparatus for Quantitative Reasoning}\label{section:apparatus}
%
%

Following \cite{DBLP:conf/lics/BacciMPP18}, we now introduce a logical apparatus for reasoning about quantitative $\qsig$--algebras.
We use the following notation to improve readability: for a set $\vdash$ of Horn clauses (${\vdash} \subseteq \clauses{\sig}(X)$) we write $\{\eqn_i\}_{i\in I} \vdash \eqn$ to denote that the Horn clause $\bigwedge_{i \in I} \eqn_i \Rightarrow \eqn$ belongs to the set $\vdash$.

\begin{definition}\label{defn:deduce}
    A \emph{quantitative theory} over $\qsig$ is a set of Horn clauses ${\vdash} \subseteq \clauses{\sig}(X)$ such that conditions \hyperlink{rules:struct}{(I)}--\hyperlink{rules:LNE}{(VI)} hold: \todo{VV:changed from "deducibility relation" to "quantitative theory", check it is consistent with the rest}

  \noindent       \hypertarget{rules:struct}{(I)} $\vdash$ is closed under the following inference rules for any $\Gamma, \Gamma' \subseteq \eqns{\sig}X$, $\eqn, \eqnb \in \eqns{\sig}X$ and substitution $\sigma: X \rightarrow \term{\sig}X$:
        \begin{gather*}
               \begin{bprooftree}
            \AxiomC{$\Gamma \vdash \eqn$}
            \RightLabel{\textsf{Sub}}
            \UnaryInfC{$\subst^*(\Gamma) \vdash \subst^*(\eqn)$}
        \end{bprooftree}\\
        \begin{bprooftree}
            \AxiomC{$\forall \eqn \in \Gamma', \Gamma \vdash \eqn$}
            \AxiomC{$\Gamma' \vdash \eqnb$}
            \RightLabel{\textsf{Cut}}
            \BinaryInfC{$\Gamma \vdash \eqnb$}
        \end{bprooftree}\quad 
        \begin{bprooftree}
            \AxiomC{$\eqn \in \Gamma$}
            \RightLabel{\textsf{Hyp}}
            \UnaryInfC{$\Gamma \vdash \eqn$}
        \end{bprooftree}
    \end{gather*}
\noindent \hypertarget{rules:eqlog}{(II)} $\vdash$ contains, for any $\op:n:L_\op\in\qsig$ and $x,y,z \in X$, the clauses:    
 \begin{align*}
        &(\textsf{Refl})& \emptyset &\vdash x = x\\
        &(\textsf{Sym})& x=y&\vdash y=x\\
        &(\textsf{Trans}) & x=y,y=z &\vdash x=z\\
        &(\textsf{App})& \{x_i= y_i \mid i \in 1,\dots, n\} &\vdash \op(\vec{x}) = \op(\vec{y})
            \end{align*}
 \noindent         \hypertarget{rules:fuzzy}{(III)}  $\vdash$ contains, for any  $x,y\in X,\varepsilon'\geq \varepsilon, \varepsilon_i \in [0,1]$, the clauses:
                \begin{align*}
        &(\textsf{1-bdd})& \emptyset &\vdash x=_1 y\\
        &(\textsf{Max})& x=_\varepsilon y &\vdash x=_{\varepsilon'} y\\
        &(\textsf{Arch})& \{x=_{\varepsilon_i} y \mid i \in I\} &\vdash x=_{\inf\{\varepsilon_i\mid i \in I\}} y
    \end{align*}
  \noindent           \hypertarget{rules:comp}{(IV)}  $\vdash$ contains, for any  $x,y,z\in X$, $\varepsilon \in [0,1]$, the clauses:
        \begin{align*}
        &(\textsf{Comp}_\ell)& x=y,x=_\varepsilon z &\vdash y=_\varepsilon z\\
        &(\textsf{Comp}_r)& x=y,z=_\varepsilon x &\vdash z=_\varepsilon y
    \end{align*}
 \noindent        \hypertarget{rules:gmet}{(V)} depending on the notion of $\GMet$ used, $\vdash$ contains an appropriate subset of the following clauses for any $x,y,z \in X$, $\varepsilon, \varepsilon^\prime \in [0,1]$:
    \begin{align*}
        x=_\varepsilon y &\vdash y=_\varepsilon x &&\eqref{eq:symm}\\
        \emptyset&\vdash x=_0 x&&\eqref{eq:refl}\\
        x=_0 y&\vdash x= y&&\eqref{eq:idofind}\\
        x=_{\varepsilon} y, y=_{\varepsilon'} z&\vdash x=_{\varepsilon+\varepsilon'} z&&\eqref{eq:trineq}\\
        x=_{\varepsilon} y, y=_{\varepsilon'} z&\vdash x=_{\max\{\varepsilon,\varepsilon'\}} z&&\eqref{eq:strongtrineq}
    \end{align*}

   \noindent      \hypertarget{rules:LNE}{(VI)} $\vdash$ is closed under the following inference rule, for any $\op:n:L_\op\in\qsig$ and for any set $\vec{x}\cup\vec{y}$ = $\{x_1,\dots,x_n,y_1,\dots,y_n\}$ of up to $2n$ variables (not necessarily distinct):
    \begin{gather*}
        \begin{bprooftree}
            \AxiomC{$(\vec{x}\cup \vec{y},\Delta) \in \GMet$}
            \AxiomC{$\delta = L_\op(\Delta)(\vec{x},\vec{y})$}
            \RightLabel{$L$--\textsf{NE}}
            \BinaryInfC{$\left\{ w =_{\Delta(w,z)} z \mid w,z \in \vec{x}\cup\vec{y} \right\} \vdash \op(\vec{x}) =_\delta \op(\vec{y})$}
        \end{bprooftree}
    \end{gather*}

 \end{definition}
 
Condition \hyperlink{rules:struct}{(I)} is standard and reflects the semantics of $\vdash$ as a theory of universally quantified implications. Condition \hyperlink{rules:eqlog}{(II)} includes the standard axioms of equational logic, thus \hyperlink{rules:struct}{(I)}+\hyperlink{rules:eqlog}{(II)} allows to perform equational reasoning regarding equations ($s=t$). Condition \hyperlink{rules:fuzzy}{(III)} poses the constraints on quantitative equations ($s=_\varepsilon t$) ensuring the intended semantics: $d(s,t)\leq \varepsilon$, for any fuzzy relation $d\in\FRel$. Condition \hyperlink{rules:comp}{(IV)} adds two axioms governing the logical interplay between equality ($=$) and the quantitative relations $=_\varepsilon$. It expresses the fact that equality is a congruence relation (both on the left and the right argument) for the relation $=_\varepsilon$, for all $\varepsilon \in [0,1]$. Condition \hyperlink{rules:gmet}{(V)} adds to the deductive system the implicational axioms defining each category of generalised metric spaces $\GMet$. Finally, Condition \hyperlink{rules:LNE}{(VI)} expresses the property that, for any $\op:n:L_\op\in\qsig$, the operation $\op$ is $L_\op$--nonexpansive. The Horn clause introduced has up to $(2n)^2$ premises: quantitative equations of the form $w=_{\Delta(w,z)} z$, where $\Delta(w,z)$ is a number in $[0,1]$, for each choice of $w,z \in \vec{x}\cup\vec{y}$. We can see these numbers as defining a fuzzy relation $\Delta:  \vec{x}\cup\vec{y}\rightarrow [0,1]$. The proviso requires that  $( \vec{x}\cup\vec{y}, \Delta)$ is a $\GMet$ space. This is therefore a constraint on the ${(2n)}^2$ values $\Delta(w,z)$. If the proviso is satisfied, since $L$ is a $\GMet$ lifting, $L_\op( \vec{x}\cup\vec{y}, \Delta)$ is a $\GMet$ space too and the value in the quantitative equation in the conclusion (i.e., $\delta= L_\op(\Delta)(\vec{x},\vec{y})$) is defined.

\begin{example}
\label{example_LNE}
In order to improve readability when displaying instances of the $L$--\textsf{NE} rule, we will often omit some of the $(2n)^2$ premises $w=_{\Delta(w,z)} z$ when $\Delta(w,z)$ is implicitly understood from the context. For instance, consider the case $\GMet=\Met$ and a binary operation $\op:2:L_\times$ with $L_\times$ the $\sup$--product lifting. An instance of the $L$--NE rule is:
$$
            x_1 =_{\varepsilon_1} y_1, x_2 =_{\varepsilon_2} y_2 \vdash \op(x_1,x_2) =_{\max\{\varepsilon_1,\varepsilon_2\}}\op(y_1,y_2)
$$
thus implicitly assuming all other premises to be of the form $w =_1 z$ if $w\neq z$, and $w=_0z$ otherwise, for $w,z\in \vec{x}\cup\vec{y}$. The fuzzy relation on $\vec{x}\cup\vec{y}$ described by these premises is therefore:
\begin{equation*}
    \begin{tikzcd}
        {x_1} \arrow["0" description, no head, loop, distance=2em, in=215, out=145]& {y_1} \arrow["0" description, no head, loop, distance=2em, in=35, out=325] \\
        {x_2}\arrow["0" description, no head, loop, distance=2em, in=215, out=145] & {y_2}\arrow["0" description, no head, loop, distance=2em, in=35, out=325] 
        \arrow["1"', no head, from=1-1, to=2-1]
        \arrow["{\varepsilon_1}", no head, from=1-1, to=1-2]
        \arrow["1", no head, from=1-2, to=2-2]
        \arrow["{\varepsilon_2}"', no head, from=2-1, to=2-2]
        \arrow["1"{description, pos=0.2}, no head, from=1-1, to=2-2]
        \arrow["1"{description, pos=0.2}, no head, from=2-1, to=1-2]
    \end{tikzcd}
\end{equation*}
which indeed satisfies the axioms of $\Met$. Since we are considering the $\sup$--product lifting,  $\delta=L_\times(\Delta)(( x_1,x_2), ( y_1,y_2)) = \max\{\varepsilon_1,\varepsilon_2\}$. Hence all provisos of the $L$--\textsf{NE} rule are satisfied, meaning that this is a valid instance of the $L$--\textsf{NE} rule.

\end{example}
 
\begin{definition}
    Given a set of clauses $S \subseteq \clauses{\sig}(X)$, we let $\relgen{S}$ denote the smallest quantitative theory containing $S$, and refer to it as the $\GMet$ \emph{quantitative theory axiomatised by $S$}. 
\end{definition}
%
    
\todo{RS:removed paragraph about proof trees.}

Our first main result is the soundness theorem, stating that if a Horn clause $H$ is derivable in the deductive apparatus from an axiom set $S$ of Horn clauses (i.e., $H$ is in the quantitative theory axiomatised by $S$), then indeed $H$ holds true in any $\qsig$ algebra satisfying the axioms $S$. 

\begin{theorem}[Soundness]\label{thm:soundness}
    Let $\algA = (A,d,\sem{\qsig}) \in \Alg(\qsig,S)$ and $H \in {\relgen{S}}$. Then $\algA \satisfies H$.
\end{theorem}
\begin{proof}
We show that each rule in Definition \ref{defn:deduce} is valid in $\algA$. 

\hyperlink{rules:struct}{(I)} The inference rules \textsf{Sub}, \textsf{Cut} and \textsf{Hyp} are valid by purely logical arguments, as the semantics of clauses $\{\phi_i \}_{i\in I}\relgen{S} \phi$ are universally quantified implications: $\forall{\vec{x}}.\big(\bigwedge_{i \in I} \eqn_i \Rightarrow \eqn\big)$.

\hyperlink{rules:eqlog}{(II)} The clauses \textsf{Refl}, \textsf{Sym}, \textsf{Trans} and \textsf{App} are valid because equality ($=$) is an equivalence relation and  is trivially compatible with all operations $\op\in\qsig$ (i.e., it is a congruence).  

\hyperlink{rules:fuzzy}{(III)} The clauses \textsf{1-bdd}, \textsf{Max} and \textsf{Arch} are valid because the distance $d$ has type $d: A \times A \rightarrow [0,1]$ and the interpretation of quantitative equations $x=_\varepsilon y$ is $d(\assign(x),\assign(y))\leq\varepsilon$ for any variable assignment $\assign: X \rightarrow A$. 

%

\hyperlink{rules:comp}{(IV)} The rules $\textsf{Comp}_\ell$ and $\textsf{Comp}_r$ are valid because equality ($=$) is trivially a congruence for all relations $=_\varepsilon$.

\hyperlink{rules:gmet}{(V)} The clauses corresponding to axioms in $\GMet$ are valid because $\algA = (A,d,\sem{\qsig})$ is in $\GMet$, and
 the interpretation of the Horn clauses (universally quantified implications) coincides with the axioms  of $\GMet$ as stated in Section \ref{sec:gmets}.

 \hyperlink{rules:LNE}{(VI)} The inference $L$--\textsf{NE} has a proviso ($(\vec{x}\cup\vec{y},\Delta) \in \GMet$) stating that the finite set $\vec{x}\cup\vec{y}$ of variables endowed with the distances $\Delta(w,z)$, for $w,z\in\vec{x}\cup\vec{y}$,  is a $\GMet$ space. Assume this as hypothesis. Since $\op:n:L_\op\in\qsig$, we know that $L_\op$ is a lifting on $\GMet$. Therefore the set
 $(\vec{x}\cup\vec{y})^n$ equipped with the distance $L_\op(\Delta)$ is an element of $\GMet$. Hence, the numerical value
$$\delta =  L_\op(\Delta)\left(( x_1\dots, x_n),( y_1,\dots, y_n) \right)$$
is defined. We need to prove that the Horn clause
\begin{equation}\label{eq_to_prove_soundness:6}
\left\{ w =_{\Delta(w,z)} z \mid w,z \in \vec{x}\cup\vec{y} \right\} \vdash \op(\vec{x}) =_{\delta} \op(\vec{y})
\end{equation}
holds in $\algA$. Let $\assign: X \rightarrow A$ be an assignment and assume that, for all $w,z \in \vec{x}\cup\vec{y}$, 
$$d(\sem{w}^{\assign},\sem{z}^{\assign})\leq \Delta(w,z)$$
holds. As consequence, the map \[f: (\vec{x}\cup\vec{y},\Delta) \rightarrow (A,d) = w \mapsto \sem{w}^{\assign}\] 
is nonexpansive. Thus, the lifting 
$$L_\op(f): L_{\op}(\vec{x}\cup\vec{y},\Delta) \rightarrow L_{\op}(A,d)$$
i.e., 
\begin{equation} \label{non_ex_property:C}
L_\op(f): \big((\vec{x}\cup\vec{y})^n,L_{\op}(\Delta)\big) \rightarrow \big(A^n,L_\op(d)\big)
 \end{equation}
 is also nonexpansive, and we have the following derivation which implies (\ref{eq_to_prove_soundness:6}), i.e., the validity of the conclusion:
        \begin{align*}
            &d(\sem{\op(\vec{x})}^{\assign},\sem{\op(\vec{y})}^{\assign}) & \\
            &= d\left(\sem{\op}(\sem{{x}_1}^{\assign},\dots,\sem{{x}_n}^{\assign}),\sem{\op}(\sem{{y}_1}^{\assign},\dots,\sem{{y}_n}^{\assign})\right) & \\
            &\leq L_{\op}(d)\left( (\sem{{x}_1}^{\assign},\dots,\sem{{x}_n}^{\assign}),(\sem{{y}_1}^{\assign},\dots,\sem{{y}_n}^{\assign}) \right) & \mathrm{(A)}\\
            &= L_{\op}(d)(L_{\op}(f)(\vec{x}),L_{\op}(f)(\vec{y})) & \mathrm{(B)}\\
            &\leq L_{\op}(\Delta)(\vec{x},\vec{y}) & \mathrm{(C)}\\ 
                   &= \delta
        \end{align*}
        where $(A)$ applies the fact that $\sem{\op}$ is $L_\op$--nonexpansive, (B) follows as $L_{\op}(f)$ applies $f$ pointwise to $n$-ary tuples, and $(C)$ uses nonexpansiveness of $L_{\op}(f)$ from (\ref{non_ex_property:C}).
 \end{proof}

%% file: results.tex



Given a lifted signature $\qsig$ and a set of Horn clauses $S$ axiomatising a theory $\vdash_S$, we describe in Subsection \ref{subsection_term} the construction of the term $(\qsig,S)$--algebra (denoted $\qterm{\qsig,S}(A,d)$) on a given $\GMet$ space $(A,d)$. We then show (Theorem \ref{thm:qtermmonad}) how this yields a monad $\qterm{\qsig,S}$ on $\GMet$.\todo{RS: changed}

Next, in Subsection \ref{subsection_free} we show two main results regarding this monad. First (Theorem \ref{term_is_free:thm}), for any given $(A,d)\in\GMet$, the algebra $\qterm{\qsig,S}(A,d)$ is the free algebra in $\Alg(\qsig,S)$ generated by $(A,d)$. Second (Theorem \ref{EM:theorem}), if all Horn clauses in $S$ are basic (see Definition \ref{def:hornclauses}), then $\Alg(\qsig,S)\cong \EM(\qterm{\qsig,S})$.

The definition of the monad $\qterm{\qsig,S}$ and the proof techniques used to establish the two theorems are inspired by those of \cite{radu2016,  DBLP:conf/lics/BacciMPP18}. In fact, the latter can be seen as special instances, in our framework, when $\GMet = \Met$ and all liftings in $\qsig$ are $\sup$--product liftings.
\subsection{The Term Monad}
\label{subsection_term}
Fix a quantitative theory $\vdash$ over a lifted signature $\qsig$. The construction of the term monad is done via several steps. First, we consider the set of ground terms, i.e., the set of terms without variables ($\term{\sig}\emptyset$) and we define on them a congruence $\termeq{\vdash}$ and a fuzzy relation $d_{\vdash}$ induced by the equations and quantitative equations in $\vdash$. We then show in Lemma \ref{equation_extraction_lemma} how these allow us to build a quantitative $\qsig$--algebra over quotiented $\term{\sig}\emptyset$ terms.
\begin{definition}\label{equation_extraction}
We let $\mathrm{E}(\vdash)$ (resp. $\mathrm{QE}(\vdash)$) be the set of equations (resp. quantitative equations) over $\term{\sig}\emptyset$ that are conclusions of Horn clauses $H\in\ \vdash$ having no premises. Formally:\todo{RS: fixed the theory}
\begin{align*}
&\mathrm{E}(\vdash) = \{ s=t \mid \emptyset\vdash s=t  \text{, for } s,t\in \term{\sig}\emptyset \}\\
&\mathrm{QE}(\vdash) = \{ s=_\varepsilon t \mid \emptyset\vdash s=_\varepsilon t  \text{, for } s,t\in\term{\sig}\emptyset  \}.
\end{align*}
Based on these, we define the following relation and fuzzy relation over $\term{\sig}\emptyset$:
$$\termeq{\vdash}\subseteq \term{\sig}\emptyset\times \term{\sig}\emptyset \ \ \ \  s \termeq{\vdash} t \Leftrightarrow (s,t)\in \mathrm{E}(\vdash)$$
$$
d_{\vdash}: \term{\sig}\emptyset\times \term{\sig}\emptyset\rightarrow [0,1] \ \ \  d_{\vdash}(s,t) = \inf\left\{ \varepsilon \mid  s=_\varepsilon t \in \mathrm{QE}(\vdash) \right\}.$$
\end{definition}

\begin{lemma}\label{equation_extraction_lemma} \todo{VV: changed statement of the lemma and modified/expanded the proof, check}
The following hold:
\begin{enumerate}
\item\label{termlemma:eqcong} The relation $\termeq{\vdash}$ is an equivalence relation on $\sig$--terms without variables and is compatible with all operations.
\item\label{termlemma:freesetalg} $(\term{\sig}\emptyset/{\termeq{\vdash}}, \sem{\sig})$ is the free $(\sig, \termeq{\vdash})$--algebra on the empty set, with carrier $\term{\sig}\emptyset/{\termeq{\vdash}}$ and operations $\sem{\sig}$:\todo{RS:removed up to iso}
$$\sem{\op}([t_1]_{\termeq{\vdash}},\dots, [t_n]_{\termeq{\vdash}}) = [\op(t_1,\dots,t_n)]_{\termeq{\vdash}}.$$
\item\label{termlemma:fuzzyrel} The fuzzy relation $d_{\vdash}$ satisfies the following properties:
\begin{enumerate}
\item\label{termlemma:fuzzyrel:arch}  $d_{\vdash}(s,t) \leq\varepsilon $ if and only if $ (s=_\varepsilon t)\in \mathrm{QE}(\vdash)$
\item\label{termlemma:fuzzyrel:comp}  $d_{\vdash}$ preserves the equivalence $\termeq{\vdash}$, i.e., $d_{\vdash}$ is well defined on $\termeq{\vdash}$--equivalence classes: $$d_{\vdash} :   {\term{\sig}\emptyset}/{\termeq{\vdash}} \times  {\term{\sig}\emptyset}/{\termeq{\vdash}}\rightarrow [0,1].$$
\end{enumerate}
\item\label{termlemma:termdgmet} $({\term{\sig}\emptyset/{\termeq{\vdash}}}, d_{\vdash})$ is a $\GMet$ space.
\item\label{termlemma:termalgisaglg}  $({\term{\sig}\emptyset/{\termeq{\vdash}}}, d_{\vdash}, \sem{\sig})$ is a  quantitative $\qsig$--algebra.
\end{enumerate}
\end{lemma}
\begin{proof} 
All points are enforced by the presence of certain rules and clauses in the syntactic proof system, and the fact that $\vdash$, being a theory, is closed under them. Item \ref{termlemma:eqcong} follows by $\textsf{Refl}$, $\textsf{Sym}$,  $\textsf{Trans}$ and $ \textsf{App}$. Item \ref{termlemma:freesetalg} follows from the characterisation of free $\sig$--algebras from Subsection \ref{subsection_universal}. Item \ref{termlemma:fuzzyrel:arch} follows from $\textsf{Max}$ and  $\textsf{Arch}$, and \ref{termlemma:fuzzyrel:comp} from $ \textsf{Comp}_\ell$ and $\textsf{Comp}_r$. Item \ref{termlemma:termdgmet} follows from the axioms in \hyperlink{rules:gmet}{(V)} corresponding to $\GMet$.\todo{RS:changed point to item and added labels and refs.}

Lastly, Item \ref{termlemma:termalgisaglg} is enforced by the $L$--\textsf{NE} rule. We discuss this case in greater detail. We need to show that, for  any $\op:n:L_\op\in\qsig$, the interpretation $\sem{\op}$
$$\sem{\op}([t_1]_{\termeq{\vdash}},\dots, [t_n]_{\termeq{\vdash}}) = [\op(t_1,\dots,t_n)]_{\termeq{\vdash}}$$
 is $L_\op$--nonexpansive. This means checking that
 $$\sem{\op}: \left(({\term{\sig}\emptyset/{\termeq{\vdash}}})^n, L_\op(d_{\vdash})\right)
 \rightarrow \left({\term{\sig}\emptyset}/{\termeq{\vdash}}, d_{\vdash}\right)
 $$
 is nonexpansive, i.e., that for any $\vec{s}=(s_1,\dots, s_n)$ and $\vec{t}= (t_1,\dots, t_n)$ in $({\term{\sig}\emptyset})^n$, 
\begin{equation}\label{op_Lopd}
d_{\vdash}(\op(\vec{s}),\op(\vec{t})) \leq L_\op(d_{\vdash})(\vec{s},\vec{t})
 \end{equation}
Using the \textsf{Sub} rule, we instantiate the $L$--\textsf{NE} rule with premises 
\[p =_{\Delta(p,q)} q \quad \text{for } p,q \in  \{s_1,\dots, s_n\}\cup\{t_1,\dots, t_n\}\]
where $\Delta(p,q)=d_{\vdash}(p,q)$. This set of premises satisfies the proviso of the $L$--\textsf{NE} rule, since $d_{\vdash}$ is a $\GMet$ relation (Item \ref{termlemma:termdgmet}) because all premises are in $\mathrm{QE}(\vdash)$ (Item \ref{termlemma:fuzzyrel:arch}). Hence, also the quantitative equation in the conclusion of the $L$--\textsf{NE} rule is in $\mathrm{QE}(\vdash)$ (apply \textsf{Cut}):
\begin{equation}\label{eq:LNEconclusion}
    \op(\vec{s}) =_{L_\op(\Delta)(\vec{s},\vec{t})} \op(\vec{t}) \in \mathrm{QE}(\vdash)
\end{equation}

Now, since we have the isometric embedding
\[\left((\{s_1,\dots, s_n\}\cup\{t_1,\dots, t_n\}), \Delta\right)\hookrightarrow ({\term{\sig}\emptyset}, d_{\vdash})\]
and $L_\op$ preserves isometric embeddings, this implies 
$$
L_\op(\Delta)(\vec{s},\vec{t})=L_\op(d_{\vdash})(\vec{s},\vec{t}),$$
and thus by Item \ref{termlemma:fuzzyrel:arch} and \eqref{eq:LNEconclusion}, we conclude \eqref{op_Lopd} holds.
\end{proof}

We remark that the last step of this proof uses the technical assumption that liftings $L_\op$ preserve isometric embeddings. In contrast, the proof of Theorem \ref{thm:soundness} (soundness) can be carried out without this hypothesis. 
Therefore, this technical assumption is not needed to reason syntactically about equality and distance in quantitative algebras but is required to ensure that the construction of the term algebra (à la \cite{radu2016}) is valid. It is also used in the proof of Theorem \ref{EM:theorem}.\looseness=-1

\todo{check that we introduce "term algebra", and possibly remove this ground term algebra terminology, as not used?}

Now, given a $\GMet$ space $(A,d)$, we aim at defining a $\qsig$--algebra over terms generated from $A$ (i.e., $\term{\sig}A$ instead of $\term{\sig}\emptyset$), taking into account the distance on $A$ given by $d$. We do so via an extension of the theory $\vdash$.
 
 \begin{definition}[Theory Extension]
Given a $\GMet$ space $(A,d)$, a lifted signature $\qsig$ and a theory $\vdash$ over $\qsig$, we define:
\begin{itemize}
\item a new lifted signature $$\qsig_A = \qsig \cup \{a:0:L_a\mid a \in A\},$$ where we add a fresh constant $a$ (of arity $0$) for each element $a\in A$ where $L_a = L_\times$, is the $0$--ary $\sup$--product lifting from Example \ref{example_1}. Note that we can identify $\term{\sig}A$ ($\sig$--terms with variables in $A$) with $\term{\sig_A}\emptyset$ ($\sig_A$--terms without variables).
\item a new theory $\vdash_A$ over $\qsig_A$ defined as the $\GMet$ theory generated by the set of clauses
$$\  \vdash\cup \ \{\emptyset\vdash a=_{d(a,a^\prime)} a^\prime \mid (a,a^\prime) \in A\times A\},$$
i.e., all the clauses in $\vdash$ and new ones describing the distances between the new constants in $A$. 
\end{itemize}
We refer to $\qsig_A$ as the signature $\qsig$ extended by the $\GMet$  space $(A,d)$. Similarly, $\vdash_A$ is the theory $\vdash$ extended by $(A,d)$.
\end{definition}

In what follows, we fix an axiom set of Horn clauses $S$ and the associated $\qsig$--theory $\vdash_S$ axiomatised by $S$. 
Its extension by a $\GMet$ space $(A,d)$ is the theory $\vdash_{S_A}$ over $\qsig_A$ whose term algebra (as in Lemma \ref{equation_extraction_lemma}, Item 5) is
$$({\term{\sig_A}\emptyset/{\termeq{\vdash_{S_A}}}}, d_{\vdash_{S_A}}, \sem{\sig_A})$$
or, identifying $\term{\sig_A}\emptyset$ with $\term{\sig}A$, the $\qsig_A$--algebra
$$({\term{\sig}A/{\termeq{\vdash_{S_A}}}}, d_{\vdash_{S_A}}, \sem{\sig_A}).$$
We can turn this into a $\qsig$--algebra
$$({\term{\sig}A/{\termeq{\vdash_{S_A}}}}, d_{\vdash_{S_A}}, \sem{\sig})$$
by forgetting the interpretations $\sem{a}$ of all constants $a\in A$. Since the set of Horn clauses $S$ is fixed, we introduce the following shortcuts to ease the notation:
\[{\termeq{A}} :={\termeq{\vdash_{S_A}}} \qquad \qterm{\qsig,S}A :=  \term{\sig}A/{\termeq{A}} \qquad \qterm{\qsig,S}d := d_{\vdash_{S_A}}\]
so that, for any $(A,d)$ we have a $(\qsig,S)$--algebra $(\qterm{\qsig,S}A, \qterm{\qsig,S}d, \sem{\sig})$.
The assignment
$$
(A,d) \mapsto (\qterm{\qsig,S}A, \qterm{\qsig,S}d, \sem{\sig})
$$
can be turned into a functor $\qterm{\qsig,S}:\GMet\rightarrow \Alg(\qsig,S)$ by defining, for each nonexpansive map $f:(A,d)\rightarrow(B,\Delta)$, 
\todo{VV:there was $\qterm{\qsig,S}(f) = \term{\sig,\termeq{A}}(f)$, but apparently not good, check def. below:}
$$\qterm{\qsig,S}(f):= [t]_{\termeq{A}} \mapsto [\term{\sig}(f)(t)]_{\termeq{B}} $$
which is equivalent to
$$\qterm{\qsig,S}(f) :=  [t(a_1,\dots, a_n)]_{\termeq{A}} \mapsto [t(f(a_1),\dots, f(a_n))]_{\termeq{B}} 
$$

To check that $\qterm{\qsig,S}$ is indeed a functor, one needs to verify that $\qterm{\qsig,S}(f)$ is well-defined on equivalence classes, nonexpansive, and commutes with operations in $\sig$, and that $\qterm{\qsig,S}$ preserves composition. The following lemma implies the first two properties.
\begin{lemma}\label{lem:termf}
    Let $f:(A,d)\rightarrow(B,\Delta)$ be an arrow in $\GMet$. For all $s,t \in \term{\sig} A$,
    \begin{align*}
        [s]_{\termeq{A}} = [t]_{\termeq{A}} &\Rightarrow [\term{\sig}f(s)]_{\termeq{B}} = [\term{\sig}f(t)]_{\termeq{B}}\\
        \qterm{\qsig,S}d([s]_{\termeq{A}}, [t]_{\termeq{A}})\leq \varepsilon &\Rightarrow \qterm{\qsig,S}\Delta([\term{\sig}f(s)]_{\termeq{B}}, [\term{\sig}f(t)]_{\termeq{B}})\leq \varepsilon
    \end{align*}
\end{lemma}
The commutation with operations and preservation of composition follow from the fact that $\term{\sig}(f)$ commutes with the operations in $\sig$ and $\term{\sig}$ preserves composition.

Hence $\qterm{\qsig,S}:\GMet \rightarrow \Alg(\qsig,S)$ is indeed a functor. It can be turned, by application of the forgetful functor (every algebra in $ \Alg(\qsig,S)$ is a $\GMet$ space), to a functor of type
$$\qterm{\qsig,S}:\GMet \rightarrow \GMet. $$ The latter can be given the structure of a monad on $\GMet$ by defining unit  $\widehat{\eta}_{(A,d)} : (A,d) \rightarrow \qterm{\qsig,S}(A,d)$ and multiplication $\widehat{\mu}_{(A,d)} : \qterm{\qsig,S}\qterm{\qsig,S}(A,d) \rightarrow \qterm{\qsig,S}(A,d)$ as follows:
\begin{align*}
&\widehat{\eta}_{(A,d)}: a \stackrel{\widehat{\eta}_{(A,d)}}{\mapsto} [a]_{\termeq{A}}\\
&\widehat{\mu}_{(A,d)}: [{t([{t_1}]_{\termeq{A}},\dots,[{t_n}]_{\termeq{A}})}]_{\termeq{ \qterm{\qsig,S}A}} \stackrel{\widehat{\mu}_{(A,d)}}{\mapsto} [{t(t_1,\dots,t_n)}]_{\termeq{A}}
\end{align*}
\todo{VV: removed angled brackets notation and made it explicit}
It can be verified that these maps are nonexpansive and well defined, and that they satisfy the conditions in Definition \ref{monad:main_definition}. Therefore, we can state:
\begin{theorem}\label{thm:qtermmonad}
$(\qterm{\qsig,S}, \widehat{\eta},\widehat{\mu})$ is a monad on $\GMet$.
\end{theorem}

\subsection{Freeness and Isomorphism Theorems}
\label{subsection_free}

We are now ready to prove that $\qterm{\qsig,S}(A,d)$ is free. 
%
%

%
%

\begin{theorem}\label{term_is_free:thm}
    Let $(A,d) \in \GMet$ and $(B,\Delta,\sem{\qsig}) \in \Alg(\qsig,S)$. For any nonexpansive map $f:(A,d) \rightarrow (B,\Delta)$, there exists a unique $\qsig$--algebra homomorphism $f^*:\qterm{\qsig,S}A \rightarrow B$ such that $f^* \circ \widehat{\eta}_{(A,d)} = f$. We summarize the statement in \eqref{diag:freealg}.
    \begin{equation}\label{diag:freealg}
\begin{tikzcd}
	{(A,d)} & {\qterm{\qsig,S}(A,d)} & {\qterm{\qsig,S}(A,d)} & {} \\
	& {(B,\Delta)} & {(B,\Delta,\sem{\qsig})}
	\arrow["{\widehat{\eta}_{(A,d)}}", from=1-1, to=1-2]
	\arrow["f"', from=1-1, to=2-2]
	\arrow[""{name=0, anchor=center, inner sep=0}, "{f^*}", dashed, from=1-3, to=2-3]
	\arrow["{\text{in }\GMet}", shift left=5, draw=none, from=1-1, to=1-2]
	\arrow[""{name=1, anchor=center, inner sep=0}, "{f^*}", dashed, from=1-2, to=2-2]
	\arrow["{\text{in }\Alg(\qsig,S)}", shift left=5, draw=none, from=1-2, to=1-4]
	\arrow["\forget"', shorten <=17pt, shorten >=17pt, from=0, to=1]
\end{tikzcd}
    \end{equation}
\end{theorem}
\begin{proof}
    Let $E=\mathrm{E}(\vdash_{S_A})$ (see Definition \ref{equation_extraction}). We organise the proof in four steps.
    \todo{VV: changed this proof as step 3 was not correct, please recheck carefully. Also, I used $\termeq{A}$ instead of $\termeq{E}$, to be consistent with the notation used so far}
    
    Step 1. By Lemma \ref{equation_extraction_lemma}, the carrier of $\qterm{\qsig,S}(A,d)$ is $\term{\sig_A,E}\emptyset$ (equivalently: $\term{\sig}A/{\termeq{A}}$), i.e., the free $(\sig_A,E)$--algebra on $\emptyset$.
    
   Step 2.\todo{Say that it is Lemma 4.4} The algebra $(B,\Delta, \sem{\qsig}_B) \in \Alg(\qsig,S)$ can be expanded to become an algebra over the extended signature with the aid of the nonexpansive map $f:A\rightarrow B$. Namely, we interpret the added constants in $A$ as follows:
   $$
   \sem{a}_B := f(a).
   $$
   Since $f$ is nonexpansive, the expanded  $(B,\Delta, \sem{\qsig_A}_B)$ satisfies the additional clauses on constants:
   $$
   \Delta\big( \sem{a}_B , \sem{a^\prime}_B\big) = \Delta( f(a), f(a^\prime)) \leq d_A(a,a^\prime),
   $$
   and therefore it is a model of the extended theory $\vdash_{S_A}$. Hence  $(B,\Delta, \sem{\qsig_A}_B)\in \Alg(\qsig_A,S_A)$. This means that all equations in $E=\mathrm{E}(\vdash_{S_A})$ are validated in $B$. This in turn means that $(B, \sem{\sig_A}_B)$ (forgetting the metric) is a $(\sig_A,E)$--algebra. 

   Step 3. Combining the first two steps, we obtain a unique $(\sig_A,E)$--algebra homomorphism \[g^*_\emptyset: \term{\sig_A,E}\emptyset \rightarrow B\]
where $g^*_\emptyset$ is the homomorphic extension of the empty function $g_\emptyset:\emptyset \to B$. By identifying $\term{\sig_A,E}\emptyset$ with $\term{\sig}A/{\termeq{A}}$, we turn $g^*_\emptyset$ into a function of type $\term{\sig}A/{\termeq{A}} \rightarrow B$, which we denote $f^*$. By the definition of $g^*_\emptyset$ we have $f^*([ a]_{\termeq{A}})= \sem a_B=f(a)$, which implies that $f^* \circ \widehat{\eta}_{(A,d)}= f$.
%
   
   Step 4. We now conclude by proving that $f^*$ is a morphism in $\Alg(\qsig,S)$, namely, it is a $\sig$--algebra homomorphism and it is nonexpansive. The former follows from Step 3 which defined $f^*$ as $g^*_\emptyset$, which is a $(\sig_A,E)$--algebra homomorphism and thus preserves all operations in $\sig$. For the latter, take arbitrary elements  $[s]_{\termeq{A}},  [t]_{\termeq{A}}\in \qterm{\qsig,S}(A,d)$ and assume $\qterm{\qsig,S}(d)(  [s]_{\termeq{A}},  [t]_{\termeq{A}})=\varepsilon$, which means that $
\emptyset \relgen{S_A} s=_\varepsilon t$. We need to show that in $(B,\Delta, \sem{\qsig}_B) \in \Alg(\qsig,S)$ it holds:
   \begin{equation}
   \label{goal_step4}
   \Delta(f^*([s]_{\termeq{A}}), f^*([t]_{\termeq{A}})) \leq \varepsilon.
   \end{equation} 
   Since we already know that $(B,\Delta, \sem{\qsig_A}_B)\in \Alg(\qsig_A,S_A)$ (Step 2), we have that $(B,\Delta, \sem{\qsig_A}_B) \satisfies s=_\varepsilon t$, which means that 
\begin{equation}\label{eq_aux_step4a}
\Delta(\sem{s}_B,\sem{t}_B)\leq \varepsilon.
\end{equation}
 To conclude, it is sufficient to observe, using the definition of $g^*_\emptyset$, that: 
   \begin{equation}\label{eq_aux_step4b}
   f^*([t]_{\termeq{A}})=\sem{t}_B
   \end{equation} 
    where $\sem{-}_B$ is the extended interpretation to $\sig_A$. Hence, from (\ref{eq_aux_step4a}) and (\ref{eq_aux_step4b}) we derive the desired inequality (\ref{goal_step4}).

\end{proof}

We now focus our attention on the case of $\qsig$ theories $\vdash_S$ generated by a set of \emph{basic} Horn clauses (Definition \ref{def:hornclauses}), that is, of the form $\bigwedge^n_{i=1} \phi_i \Rightarrow \phi $, where each $\phi_i$ is a (quantitative) equation ($x=y$ or $x=_\varepsilon y$) between variables.
\todo{this is defined before, so possibly here we can save space}

\begin{theorem}\label{EM:theorem}
Let $\qsig$ be a lifted signature and $S$ a set of basic Horn clauses. Then $\EM(\qterm{\qsig,S}) \cong \Alg(\qsig,S)$.
\end{theorem}
\begin{proof}[Proof sketch]
    Let $(A,d,\alpha) \in \EM(\qterm{\qsig,S})$, we define the interpretations $\sem{\qsig}_{\alpha}$ as follows: for any $\op :n \in \sig$ and $\vec{a}\in A^n$,
    \[\sem{\op}_{\alpha}(\vec{a}) = \alpha([{\op(\vec{a})}])\]
    where $[t]$ stands for $[t]_{\termeq{A}}$.
    \todo{VV: removed angled brackets here}
    We claim that $(A,d,\sem{\qsig}_{\alpha}) \in \Alg(\qsig,S)$. First, we show $\sem{\op}_{\alpha}$ is $L_\op$--nonexpansive. Given $\vec{a}, \vec{b} \in L_\op(A,d)$, let $\Delta$ be the restriction of $d$ on $\vec{a}\cup \vec{b}$, we have
    \begin{align*}
        d(\alpha([{\op(\vec{a})}]), \alpha([{\op(\vec{b})}])) &\leq \qterm{\qsig,S}d([{\op(\vec{a})}],[{\op(\vec{b})}])\\
        &\leq L_\op(\Delta)(\vec{a},\vec{b})\\
        &= L_\op(d)(\vec{a},\vec{b}).
    \end{align*}
    The first inequality holds because $\alpha$ is nonexpansive, the second inequality uses the rule $L$--\textsf{NE}, and the equality is the fact that $L_\op$ preserves isometric embeddings.

    An adaptation of the argument in the proof of Theorem 4.2 in \cite{DBLP:conf/lics/BacciMPP18} shows $(A,d,\sem{\qsig}_{\alpha})$ satisfies the clauses in $S$. This defines a functor $\widehat{P}:\EM(\qterm{\qsig,S}) \rightarrow \Alg(\qsig,S)$ acting trivially on morphisms and sending $(A,d,\alpha)$ to $(A,d,\sem{\qsig}_\alpha)$.

    In the converse direction, let $\algA = (A,d,\sem{\qsig}) \in \Alg(\qsig,S)$, we define $\widehat{\alpha}_\algA: \qterm{\qsig,S}(A,d) \rightarrow (A,d)$ inductively as follows: for any $a \in A$, $\widehat{\alpha}_\algA([{a}]) = a$ and $\forall \op:n \in \sig$, $\forall t_1,\dots, t_n\in \term{\sig}A$,
    \[\widehat{\alpha}_\algA([{\op(t_1,\dots,t_n)}]) = \sem{\op}\left( \widehat{\alpha}_\algA([{t_1}]),\dots, \widehat{\alpha}_\algA([{t_n}]) \right).\]
    This defines a functor $\widehat{P}^{-1}:\Alg(\qsig,S) \rightarrow \EM(\qterm{\qsig,S})$. It actstrivially on morphisms and sends $\algA = (A,d,\sem{\qsig})$ to $(A,d,\widehat{\alpha}_\algA)$.

    The functor $\widehat{P}$ and $\widehat{P}^{-1}$ are inverses and we conclude the desired isomorphism.
\end{proof}

%% file: examples.tex


In Sections \ref{section_definitions} and \ref{section_termfree} we have introduced the new notions of lifted signatures $\qsig$ and quantitative $\qsig$--algebras, the deductive apparatus to reason about them, and we stated our main results: Theorem \ref{thm:soundness} (soundness), Theorem \ref{term_is_free:thm} (free algebras) and Theorem \ref{EM:theorem} ($\EM(\qterm{\qsig,S}) \cong \Alg(\qsig,S)$ for basic theories). We now show the applicability of our framework. 


\subsection{Applications already studied in the literature}\label{sec:examples:1}
As already pointed out, the framework of \cite{radu2016, DBLP:conf/lics/MardarePP17, DBLP:conf/lics/BacciMPP18} can  be seen as a special case of our framework when: (1) the generalised metric space $\GMet$ considered is $\Met$ and (2) all liftings in the lifted signature $\qsig$ are the $sup$--product lifting $L_\times$ (see Example \ref{example_1}). For several interesting examples of applications, more can be said.

We first recall some definitions. Given a set $A$, we let $\mD(A)$ denote the set of finitely supported probability distributions on $A$, i.e., functions  $\distr:A\rightarrow [0,1]$ such that $|\{ a\mid \distr(a) >0 \}|$ is finite. For a given $a\in A$, the Dirac distribution $\dirac{a}\in\mD(A)$ assigns $1$ to $a$, and $0$ to all other elements. 
Convex algebras are algebras for the following signature and set of axioms:
$$
 \Sigma = \{ +_p : 2\}_{p\in (0,1)} \ \ \ \ \  E = \begin{Bmatrix}x+_px  =  x,\;\;  x+_py  =  y+_{1-p}x,\\
  (x+_qy)+_pz = x+_{pq}(y+_{\frac{p(1-q)}{1-pq}}z)\end{Bmatrix} 
$$
It is well-known (see, e.g., \cite{jacobs:2010}) that $\mD(A)$ with operations defined as:
$$\sem{+_p}(\distr,\distrb) :=   a\mapsto \big(p\cdot\distr(a) + (1-p)\cdot\distrb(a)\big)$$
is (up to isomorphism) the free convex algebra on the set $A$.

The ($\Met$) quantitative theory of convex algebras from \cite{radu2016, DBLP:conf/lics/BacciMPP18}
can be formalised in our framework by taking $\GMet=\Met$, lifted signature $\qsig=\{ +_p : 2: L_\times\}_{p \in (0,1)}$
%
%
%
%
%
and as generating set of Horn clauses the axioms $E$ of convex algebras together with the clause:
$$\{x_1 =_{\varepsilon_1} y_1,  x_2 =_{\varepsilon_1} y_2\}\Rightarrow x_1 +_p x_2 =_{p\varepsilon_1 + (1-p)\varepsilon_2} y_1 +_p y_2$$
known as ``Kantorovich rule''.
Note that, since the inequality 
$$
p\varepsilon_1 + (1-p)\varepsilon_2 \leq \max\{\varepsilon_1, \varepsilon_2\} 
$$
holds for all $p,\varepsilon_1,\varepsilon_2\in [0,1]$, the Kantorovich rule strictly subsumes (using the \textsf{Max} rule) the $L$--\textsf{NE} rule for $+_p$, which only states (omitting some premises, cf. Example \ref{example_LNE}):
$$  \{x_1 =_{\varepsilon_1} y_1,  x_2 =_{\varepsilon_1} y_2\}\Rightarrow x_1 +_p x_2 =_{ \max\{\varepsilon_1, \varepsilon_2\} } y_1 +_p y_2. $$
Hence, in quantitative ($\Met$) convex algebras, the operation $\sem{+_p}$ is not merely $L_\times$--nonexpansive, as it needs to satisfy the stronger constraint of the Kantorovich rule.

Consider now, for every $p\in(0,1)$, the lifting $L_K^p$ of the binary product defined as follows:
\begin{align*}
&L_K^p: (A,d) \mapsto (A\times A,L_K^p(d)) \\
&L_K^p(d)\big( ( a_1, a_2),  (b_1, b_2)\big) = d_{K}\big(  \sem{+_p}( \dirac{a_1}, \dirac{a_2}) ,  \sem{+_p}( \dirac{b_1}, \dirac{b_2})\big)
\end{align*}
where $d_{K}$ is the well-known Kantorovich distance over distributions $\mD(A)$.
 This lifting is easily seen to preserve isometric embeddings.

Then it can be shown that the  ($\Met$) quantitative theory of convex algebras, axiomatised above, can also be presented as the theory over the lifted signature $\qsig_K=\{ +_p : 2: L_K^p\}_{p \in (0,1)}$, taking as generating set of Horn clauses only the set $E$ of axioms of convex algebras.
In other words, we have cast the Kantorovich rule as a $L$--NE rule, by choosing the appropriate lifting $L_K^p$ for every operation $+_p$.  Note that the remaining clauses are just the purely equational axioms of the ($\Set$) theory of convex algebras.

The same applies in several other interesting examples. For example, also the $(\Met)$ quantitative theory of convex semilattices of \cite{DBLP:conf/concur/MioV20,DBLP:conf/lics/MioSV21} can be presented as the $(\Met)$ quantitative theory with generating clauses just the equational axioms of convex semilattices, by choosing the appropriate liftings in the lifted signature. 

\subsection{No constraints on algebraic operations}\label{sec:examples:2}

Among the variants of the framework of   \cite{DBLP:conf/lics/BacciMPP18}  that have been considered in the literature, the work of \cite{BacciBLM18} is relevant in our discussion. Indeed, the authors have observed that certain fixed--point operations on metric spaces fail to be nonexpansive (up to the $\sup$--product lifting $L_\times$) and, as such, cannot be cast in the framework of  \cite{radu2016, DBLP:conf/lics/MardarePP17, DBLP:conf/lics/BacciMPP18}. The solution adopted in  \cite{BacciBLM18} is to drop entirely all constraints on the interpretation of the algebraic operations $\sem{\op}$ and allow arbitrary maps $\sem{\op}^n:A^n\rightarrow A$. 

This approach can be seen as a particular instance of our framework by taking $\GMet = \Met$ and using lifted signatures $\qsig$ where for all $\op:n:L_\op\in\qsig$ the lifting $L_\op$ is the ``discrete'' lifting defined as follows:
$$
L_{\op}(d)\big( ( a_1,\dots, a_n) ,  ( b_1,\dots, b_n)\big) = 
\begin{cases}
0 & \textnormal{if } \forall^n_{i=1}.\  a_i = b_i  \\
1 & \textnormal{otherwise} \
\end{cases}
$$
Indeed, with this choice of lifting, the $L$--NE rule 
    \begin{gather*}
              \begin{bprooftree}
            \AxiomC{$(\vec{x}\cup \vec{y},\Delta) \in \GMet$}
            \AxiomC{$\delta = L_\op(\Delta)(\vec{x},\vec{y})$}
            \RightLabel{$L$--\textsf{NE}}
            \BinaryInfC{$\left\{ w =_{\Delta(w,z)} z \mid w,z \in \vec{x}\cup\vec{y} \right\} \vdash \op(\vec{x}) =_\delta \op(\vec{y})$}
        \end{bprooftree}
            \end{gather*}
%
 is rendered useless, as it can always by substituted with instances of the clause $\textsf{1-bdd}$, if $\vec{x} \neq \vec{y}$, or with instances of the clause $\emptyset \vdash   \op(\vec{x})=_0  \op(\vec{y})$ (coming from the axioms of $\Met$) if $\vec{x} = \vec{y}$. 

Therefore our free algebra and isomorphism theorems from Section  \ref{section_termfree} hold for the theory developed in \cite{BacciBLM18} and for further variants that can be conceived. Such results could not be automatically derived from the original framework of  \cite{DBLP:conf/lics/BacciMPP18} only allowing for $L_\times$--nonexpansive operations.

\subsection{The  \L ukaszyk--Karmowski distance on probability distributions}\label{sec:examples:3}

In this subsection we develop our main example, already presented in the introduction: the axiomatisation of  the \L ukaszyk--Karmowski distance ($\LK{d}$) on probability distributions \cite{lukaszyk04}. The distance $\LK{d}$ has very recently found application in the field of representation learning and it is at the core of the definition of the MICo (``matching under independent couplings'') behavioural distance on Markov processes of \cite{DBLP:journals/corr/abs-2106-08229}.

Recall that a diffuse metric space $(A,d)\in\DMet$ is a set $A$ with a fuzzy relation $d:A\times A\rightarrow [0,1]$ satisfying reflexivity and triangular inequality, i.e., for all $a,b,c\in A$:
\[d(a,b) = d(b,a)  \qquad  d(a,c) \leq d(a, b)  +  d(b,c).  \]
The notion of diffuse metric has been introduced in \cite[\S 4.2]{DBLP:journals/corr/abs-2106-08229}. The following diagrams depict some diffuse metric spaces $(A,d)$ with $A$ finite.  

\begin{center}
\!\!\!\begin{tikzcd}
a \arrow["0"', no head, loop, distance=2em, in=125, out=55]
\end{tikzcd}
\begin{tikzcd}
a \arrow["\frac{1}{2}"', no head, loop, distance=2em, in=125, out=55]
\end{tikzcd}
\begin{tikzcd}
a \arrow["1"', no head, loop, distance=2em, in=125, out=55]
\end{tikzcd}
\begin{tikzcd}
a \arrow["1"', no head, loop, distance=2em, in=125, out=55] \arrow[r, "\frac{1}{2}", no head] & b \arrow["0"', no head, loop, distance=2em, in=125, out=55]
\end{tikzcd}
\begin{tikzcd}
a \arrow["0"', no head, loop, distance=2em, in=125, out=55] \arrow[r, "0", no head] & b \arrow["0"', no head, loop, distance=2em, in=125, out=55]
\end{tikzcd}
\end{center}

\begin{definition}
Let $(A,d)$ be a diffuse metric space. The \L ukaszyk--Karmowski distance is the fuzzy relation $\LK{d}$ on the set of finitely supported probability distributions $\mD (A)$ defined for any $\distr, \distrb \in \mD (A)$ as
\[\LK{d}(\distr,\distrb) = \sum_{x \in \supp{\distr}} \sum_{y \in \supp{\distrb}} \distr(x)\cdot \distrb(y)\cdot d(x,y).\]
\end{definition}
\begin{proposition}\label{prop:LKdiffuse} For any diffuse metric space $(A,d)$, the space $(\mD(A), \LK{d})$ is a diffuse metric space.
\end{proposition}

Recall from Subsection \ref{sec:examples:1} that convex algebras are algebras for the signature $\Sigma = \{ +_p : 2\}_{p\in (0,1)}$ satisfying the axioms $E$, and that the free convex algebra generated by $A$ is $\mD(A)$.
We now observe, however, that on probability distributions equipped with the \L ukaszyk--Karmowski distance the operation $\sem{+_p}$ generally fails to be nonexpansive (up to the $\sup$--product lifting $L_\times$).

\todo{VV: I redid the proof below, please recheck the numbers}
\begin{lemma}
\label{lem:plusp_not_ne}
There exists a diffuse metric space $(A,d)$ such that
the following map is not nonexpansive:
$$\sem{+_p}: (\mD(A), \LK{d}) \times (\mD(A), \LK{d}) \rightarrow (\mD(A), \LK{d})$$
\end{lemma}
\begin{proof}
Fix the $\DMet$ space $A=\{a,b\}$ with $d(a,a)=d(b,b)=\frac{1}{2}$ and $d(a,b)=d(b,a) = 1$. Take the Dirac distributions $\dirac{a},\dirac{b}\in\mD(A)$. We have $\LK{d}(\dirac{a},\dirac{a}) = \LK{d}(\dirac{b},\dirac{b}) = \frac{1}{2}$, and 
$$\LK{d}(\sem{+_{\frac 1 2}}( \dirac{a}, \dirac{b}),\sem{+_{\frac 1 2}}( \dirac{a}, \dirac{b})) = \tfrac{3}{4}.$$
Recall that $\LK{d} \times \LK{d}$ is the $\sup$--product lifting of $\LK{d}$.
Hence, $\sem{+_p}$ is not nonexpansive:
\begin{align*}
    \tfrac{1}{2} 
    &= \max\{ \LK{d}(\dirac{a},\dirac{a}), \LK{d}(\dirac{b},\dirac{b})\}\\
    &=  \LK{d} \times \LK{d} ( ( \dirac{a},\dirac{b} ),   ( \dirac{a},\dirac{b} ))\\
    &< \LK{d}(  \sem{+_{\frac 1 2}}( \dirac{a}, \dirac{b}) ,  \sem{+_{\frac 1 2}}( \dirac{a}, \dirac{b}))= \tfrac{3}{4}
\end{align*}
\end{proof}

We now introduce a new lifting $L_{\textnormal{\L K}}^p$ of the binary product ensuring that $\sem{+_p}$ is $L_{\textnormal{\L K}}^p$--nonexpansive.
\todo{removed definition}
For every $p\in(0,1)$, we define the $\DMet$ lifting of the binary product:
\begin{align*}
&L_{\textnormal{\L K}}^p: (A,d) \mapsto (A\times A, L_{\textnormal{\L K}}^p(d))\\
&L_{\textnormal{\L K}}^p(d)\big( ( a_1, a_2),  (b_1, b_2)\big)
=  \LK{d}\big(  \sem{+_p}(\dirac{a_1}, \dirac{a_2}),  \sem{+_p}(\dirac{b_1}, \dirac{b_2})\big).
\end{align*}
\begin{lemma}\label{lem:LKpliftpresiso}
The lifting $L_{\textnormal{\L K}}^p$ preserves isometric embeddings.
\end{lemma}
\begin{lemma}\label{lem:convcombLKNE}
For every $\DMet$ space $(A,d)$, the operation $\sem{+_p}: \mD(A)\times\mD(A) \rightarrow  \mD(A)$ is $L_{\textnormal{\L K}}^p$--nonexpansive.
\end{lemma}

We can then consider the following $\DMet$ lifting of the signature $\sig$ of convex algebras: 
$\qsig_{\textnormal{\L K}}:= \{ +_p:2:L_{\textnormal{\L K}}^p \}_{p\in (0,1)}$,
and the quantitative $\qsig_{\textnormal{\L K}}$--theory $\vdash_E$ generated by the set $E$ of axioms of convex algebras. In this theory the $L$--\textsf{NE} rule for $+_p$ takes the following form (omitting some premises, cf. Example \ref{example_LNE}):
$$\begin{Bmatrix}x_1=_{\varepsilon_{11}} x_1,  x_2=_{\varepsilon_{21}} x_1\\
  x_1=_{\varepsilon_{12}} y_2, y_2=_{\varepsilon_{22}} y_2
  \end{Bmatrix}
\vdash x_1+_p x_2 =_\delta y_1+_p y_2$$
with $\delta= p^2 \varepsilon_{11} + (1-p) p \varepsilon_{21} + p(1-p) \varepsilon_{12} +(1-p)^2 \varepsilon_{22}$.

By application of Theorem \ref{term_is_free:thm} we know that $\Alg(\qsig_{\textnormal{\L K}}, E)$ has free algebras on $(A,d)$, for every $\DMet$ space $(A,d)$, and that these are term algebras $\qterm{\qsig_{\textnormal{\L K}},S}(A,d)$ on which we can reason syntactically. The following theorem states that these term algebras are isomorphic to $(\mD(A),\LK{d},\sem{\sig})$, the collection of finitely supported probability distributions, with \L K distance and standard convex algebras operations.




\begin{theorem}\label{thm:freealgLK}
The free algebra in $\Alg(\qsig_{\textnormal{\L K}}, E)$ on a $\DMet$ space $(A,d)$ is $(\mD(A), \LK{d}, \sem{\sig})$.
\end{theorem}

Hence we can say that the theory $\vdash_E$ axiomatises convex algebras $(\mD(A), \sem{\sig})$ with the \L K distance.

%% file: conclusion.tex

We have presented an extension of the quantitative algebra framework of 
\cite{DBLP:conf/lics/BacciMPP18,radu2016, DBLP:conf/lics/MardarePP17, DBLP:conf/calco/BacciMPP21,DBLP:conf/lics/MardarePP21} allowing us to reason on generalised metric spaces and on algebraic operations that are nonexpansive up to a lifting. This has allowed, as an illustrative example, the axiomatisation of the  \L ukaszyk--Karmowski distance on probability distributions.

One direction of future work is to explore if, and how, the recent results developed for the framework of \cite{DBLP:conf/lics/BacciMPP18} can be adapted and generalised to our setting. For example, tensor product of theories (\cite{DBLP:conf/calco/BacciMPP21} and techniques to handle fixedpoints \cite{DBLP:conf/lics/MardarePP21}.

In another direction, one can look for further generalisations. For example, it would be interesting to investigate how our treatment of $\GMet$ compares with the general relational apparatus of \cite{DBLP:conf/calco/FordMS21} and find a way to lift their more general arities. Another interesting possibility is to consider liftings of the entire signature functor 
     \[\sig := \coprod_{\op:n\in \sig} A^n \quad \sig(f):= \coprod_{\op:n \in \sig} f^n.\] 
 rather than just liftings of each of the operations. 

From a foundational standpoint, the question of  what classes of monads (e.g., finitary ones) can be constructed as term monads for quantitative theories is still open.

Generally, we plan to look at more interesting examples to drive our research on all these topics.

%% file: appendix.tex
\etocsettocstyle{}{}
\etocsettocdepth{3}
\localtableofcontents
\subsection{Background}
\subsubsection{Additional Result on Monads}
We need an additional result on monads in the full proof of Theorem \ref{EM:theorem}.
\begin{definition}[Monad functor]
    Let $(\mon,\eta^{\mon},\mu^{\mon})$ be a monad on $\Cat$ and $(\monb,\eta^{\monb},\mu^{\monb})$ a monad on $\Catb$ . A \emph{monad functor} from $\mon$ to $\monb$ is a pair $(F,\lambda)$ comprising a functor $F: \Cat \rightarrow \Catb$ and a natural transformation $\lambda: \monb F \Rightarrow F \mon$ such that (1) $\lambda \circ \eta^{\monb}F = F\eta^{\mon}$ and (2) $\lambda \circ \mu^{\monb}F = F\mu^{\mon} \circ \lambda \mon \circ \monb \lambda$.
\end{definition}
\begin{proposition}[\cite{Street1972}]\label{prop:monadfunctorfunctor}
    Let $(F,\lambda): \mon \rightarrow \monb$ be a monad functor, then there is a functor $F-\circ \lambda: \EM(\mon) \rightarrow \EM(\monb)$ sending an $\mon$--algebra $\alpha: \mon A \rightarrow A$ to $F\alpha \circ \lambda_A: \monb(FA) \rightarrow FA$ and a morphism $f: (A,\alpha) \rightarrow (A^\prime,\alpha^\prime)$ to $Ff$.
\end{proposition}
\begin{proof}
    A lower level proof is drawn in \cite{Marsden2014}.
\end{proof}
\subsubsection{Products and Coproducts in $\GMet$}
In Section \ref{sec:gmets}, we gave the construction of products and coproducts in the category $\GMet$ with no proof nor reference to a proof. We prove this here. We first prove the base case in $\FRel$.
\begin{proposition}
    Let $\{(A_i,d_i) \mid i \in I\}$ be non-empty family of fuzzy relations $\{(A_i,d_i) \mid i \in I\}$, the product is $(\prod_{i \in I} A_i, \sup_{i \in I} d_i)$ with the usual projections and the coproduct is $\left( \coprod_{i \in I} A_i, \amalg_{i \in I} d_i \right)$ with the usual coprojections.
\end{proposition}
\begin{proof}
    \textbf{Product.} The projections $\pi_i$ are clearly nonexpansive. Let $(A,d) \xrightarrow{f_i} (A_i,d_i)$ be a family of nonexpansive maps. The universal property of the product in $\Set$ yields a unique function $!:(A,d) \rightarrow (\prod_{i \in I}A_i,\sup_{i \in I}d_i)$ such that $\pi_i \circ ! = f_i$. Now, it is enough to prove the function is nonexpansive. For any $a,b\in A$, we have \begin{align*}
        (\sup_{i \in I} d_i)(!(a), !(b) &= (\sup_{i \in I} d_i)((f_i(a))_{i \in I}, (f_i(b))_{i \in I})\\
        &= \sup_{i \in I} d_i(f_i(a),f_i(b))\\
        &\leq \sup_{i \in I} d(a,b)\\
        &= d(a,b)
    \end{align*}

    \textbf{Coproduct.} The coprojections $\kappa_i$ are clearly nonexpansive (they are in fact isometries). Let $(A_i,d_i) \xrightarrow{f_i} (A,d)$ be a family of nonexpansive maps. The universal property of the coproduct in $\Set$ yields a unique function $!:(\coprod_{i \in I}A_i, \amalg_{i \in I}d_i) \rightarrow (A,d)$ such that $! \circ \kappa_i = f_i$. Now, it is enough to prove the function is nonexpansive. For any $a \in A_j$ and $b \in A_k$, if $j\neq k$, \[(\amalg_{i \in I}d_i)(a,b) = 1 \geq d(!(a),!(b)).\]
    If $j = k$,
    \[(\amalg_{i \in I}d_i)(a,b) = d_j(a,b) \geq d(f_j(a),f_j(b)) = d(!(a),!(b)).\]
\end{proof}
Now, we prove that if each fuzzy relation $(A_i, d_i)$ satisfies an axiom of \eqref{eq:symm}--\eqref{eq:strongtrineq}, then the product and coproduct satisfy that axiom. It follows that (co)products in $\GMet$ exist and are computed just like those in $\FRel$.
\begin{proposition}
    Fix a subset $G$ of the axioms \eqref{eq:symm}--\eqref{eq:strongtrineq} and let $\{(A_i,d_i) \mid i \in I\}$ be non-empty family of fuzzy relations. If every $(A_i,d_i)$ satisfies the axioms in $G$, then the product and the coproduct satisfy the axioms in $G$.
\end{proposition}
\begin{proof}
    \textbf{Product.} We proceed with each axiom independently: we suppose each $(A_i,d_i)$ satisifies it and show $(\prod_{i \in I}A_i,\sup_{i \in I}d_i)$ also satisfies it. 
    \begin{enumerate}
        \item[\eqref{eq:symm}] For any $\vec{a},\vec{b} \in \prod_{i \in I}$, since $d_i(\vec{a}_i,\vec{b}_i) = d_i(\vec{b}_i,\vec{a}_i)$ for all $i\in I$, the two sets $\left\{ d_i(\vec{a}_i,\vec{b}_i) \mid i \in I \right\}$ and $\left\{ d_i(\vec{b}_i,\vec{a}_i) \right\}$ are equal, and so are their supremums. We conclude \[(\sup_{i \in I}d_i)(\vec{a},\vec{b}) = (\sup_{i \in I}d_i)(\vec{b},\vec{a}).\]
        \item[\eqref{eq:refl}] For any $\vec{a} \in \prod_{i\in I}A_i$, since $d_i(\vec{a}_i,\vec{a}_i) = 0$ for all $i\in I$, the two sets $\left\{ d_i(\vec{a}_i,\vec{a}_i) \mid i \in I \right\}$ and $\{0\}$ are equal, and so are their supremums. We conclude \[(\sup_{i \in I}d_i)(\vec{a},\vec{a}) = 0.\]
        \item[\eqref{eq:idofind}] For any $\vec{a},\vec{b} \in \prod_{i \in I}$, if $(\sup_{i \in I}d_i)(\vec{a},\vec{b}) =0$, we have 
        \[\forall i \in I, d_i(\vec{a}_i,\vec{b}_i) \leq (\sup_{i \in I}d_i)(\vec{a},\vec{b}) =0,\]
        which implies $\forall i \in I, \vec{a}_i = \vec{b}_i$. We conclude $\vec{a} = \vec{b}$.
        \item[\eqref{eq:trineq}] For any $\vec{a},\vec{b},\vec{c} \in \prod_{i \in I}$, we have 
        \begin{equation}\label{eq:suptrineq}
            d_i(\vec{a}_i,\vec{c}_i) \leq  d_i(\vec{a}_i,\vec{c}_i) + d_i(\vec{b}_i,\vec{c}_i),
        \end{equation}
        and using standard properties of the supremum, we obtain
        \begin{align*}
            (\sup_{i \in I}d_i)(\vec{a},\vec{c}) &= \sup_{i \in I}d_i(\vec{a}_i,\vec{c}_i)\\
            &\leq \sup_{i \in I}\left( d_i(\vec{a}_i,\vec{b}_i) + d_i(\vec{b}_i,\vec{c}_i) \right)\\
            &\leq \sup_{i \in I}d_i(\vec{a}_i,\vec{b}_i) + \sup_{i \in I}d_i(\vec{b}_i,\vec{c}_i)\\
            &= (\sup_{i \in I}d_i)(\vec{a},\vec{b})+ (\sup_{i \in I}d_i)(\vec{b},\vec{c})
        \end{align*}
        \item[\eqref{eq:strongtrineq}] For any $\vec{a},\vec{b},\vec{c} \in \prod_{i \in I}$, we have 
        \begin{equation}\label{eq:supstrongtrineq}
            d_i(\vec{a}_i,\vec{c}_i) \leq  \max\{d_i(\vec{a}_i,\vec{c}_i),d_i(\vec{b}_i,\vec{c}_i)\},
        \end{equation}
        and using standard properties of the supremum, we obtain
        \begin{align*}
            (\sup_{i \in I}d_i)(\vec{a},\vec{c}) &= \sup_{i \in I}d_i(\vec{a}_i,\vec{c}_i)\\
            &\leq \sup_{i \in I}\left( \max\{d_i(\vec{a}_i,\vec{b}_i) , d_i(\vec{b}_i,\vec{c}_i)\}\right)\\
            &\leq \max\left\{ \sup_{i \in I}d_i(\vec{a}_i,\vec{b}_i) , \sup_{i \in I}d_i(\vec{b}_i,\vec{c}_i) \right\}\\
            &= \max\left\{ (\sup_{i \in I}d_i)(\vec{a},\vec{b}), (\sup_{i \in I}d_i)(\vec{b},\vec{c}) \right\}
        \end{align*}
    \end{enumerate}
    \textbf{Coproduct.} We proceed with each axiom independently: we suppose each $(A_i,d_i)$ satisifies it and show $(\coprod_{i \in I}A_i,\amalg_{i \in I}d_i)$ also satisfies it.
    \begin{enumerate}
        \item[\eqref{eq:symm}] For any $a \in A_j$ and $b \in A_k$, if $j\neq k$, then \[(\amalg_{i \in I}d_i)(a,b) = 1 = (\amalg_{i \in I}d_i)(b,a),\] otherwise if $j=k$,
        \[(\amalg_{i \in I}d_i)(a,b) = d_j(a,b) = d_j(b,a) = (\amalg_{i \in I}d_i)(b,a).\]
        \item[\eqref{eq:refl}] For any $a \in A_j$, we have 
        \[(\amalg_{i \in I}d_i)(a,a) = d_j(a,a) = 0.\]
        \item[\eqref{eq:idofind}] For any $a \in A_j$ and $b \in A_k$, if $j\neq k$, then \[(\amalg_{i \in I}d_i)(a,b) = 1 \neq 0,\] otherwise if $j=k$,
        \[(\amalg_{i \in I}d_i)(a,b) = 0 \implies d_j(a,b) = 0 \implies a=b.\]
        \item[\eqref{eq:trineq}] For any $a \in A_j$ and $b \in A_k$, $c \in A_\ell$, if either $j\neq k$ or $k \neq \ell$, then \[(\amalg_{i \in I}d_i)(a,c) \leq  1\leq (\amalg_{i \in I}d_i)(a,b) + (\amalg_{i \in I}d_i)(b,c).\] otherwise if $j=k$ and $k= \ell$, then $j = \ell$, thus
        \begin{align*}
            (\amalg_{i \in I}d_i)(a,c) &= d_j(a,c)\\
            &\leq d_j(a,b) + d_j(b,c)\\
            &= (\amalg_{i \in I}d_i)(a,b)+ (\amalg_{i \in I}d_i)(b,c).
        \end{align*}
        \item[\eqref{eq:strongtrineq}] For any $a \in A_j$ and $b \in A_k$, $c \in A_\ell$, if either $j\neq k$ or $k \neq \ell$, then \[(\amalg_{i \in I}d_i)(a,c) \leq  1\leq \max\left\{  (\amalg_{i \in I}d_i)(a,b) ,(\amalg_{i \in I}d_i)(b,c) \right\}.\] otherwise if $j=k$ and $k= \ell$, then $j = \ell$, thus
        \begin{align*}
            (\amalg_{i \in I}d_i)(a,c) &= d_j(a,c)\\
            &\leq \max\left\{ d_j(a,b) + d_j(b,c) \right\}\\
            &= \max\left\{(\amalg_{i \in I}d_i)(a,b), (\amalg_{i \in I}d_i)(b,c)\right\}.
        \end{align*}
    \end{enumerate}
\end{proof}

\subsection{Proofs of Section \ref{section_termfree}}
\subsubsection{Proof of Theorem \ref{thm:qtermmonad}}
We divide the proof in multiple lemmas, the first being a technical lemma. In short, it states that if $f:(A,d) \rightarrow (B,\Delta)$ is nonexpansive, then any $(\qsig,S)$--algebra with carrier $(B,\Delta)$ can be extended to a $(\qsig_A,S_A)$--algebra.\footnote{This result was implicitly used in the proof of Theorem \ref{term_is_free:thm}.}
\begin{lemma}\label{lem:nonexpextendedalg}
    Let $\algB:= (B,\Delta,\sem{\qsig}_B)$ be a $(\qsig,S)$--algebra. For any nonexpansive map $f:(A,d) \rightarrow (B,\Delta)$, there is a $(\qsig_A,S_A)$--algebra $(B,\Delta,\sem{\qsig_A}_{B,f})$ such that for any $t \in \term{\sig}A = \term{\sig_A}\emptyset$, $\sem{t}^{\assign}_{B,f} = \sem{\term{\sig}f(t)}^{\id_B}_B$ with $\assign:\emptyset \rightarrow B$ being the only possible assignment ($\id_B$ is omitted in the sequel).
\end{lemma}
\begin{proof}
    Setting $\sem{\op}_{B,f} = \sem{\op}_B$ for every $\op \in \sig$ and $\sem{a}_{B,f} = f(a)$ for every $a \in A$, we get all the interpretations in $\sem{\qsig_A}_{B,f}$. We write $\algB_f := (B,\Delta, \sem{\qsig_A}_{B,f})$. Note that $\algB_f$ still satisfies the clauses in $S$ as they do not involve the constants from $A$ and $\algB$ satisfied them. Moreover, since $f$ is nonexpansive,
    \[\Delta\left( \sem{a}_B , \sem{a^\prime}_B \right) = \Delta( f(a), f(a^\prime)) \leq d(a,a^\prime),\]
    thus $\algB_f$ satisfies the additional clauses on constants ($\emptyset \Rightarrow  a=_{d(a,a')} a'$) that belong to $S_A$. We conclude $\algB_f \in \Alg(\qsig_A,S_A)$.

    We proceed by induction for the last part of the lemma. If $t = a \in A$, we have $\sem{t}^{\assign}_{B,f} = f(a) = \sem{f(a)}_B$. If $t = \op(t_1,\dots, t_n)$ and we assume $\sem{t_i}^{\assign}_{B,f} = \sem{\term{\sig}f(t_i)}_B$ for each $1\leq i\leq n$, we have 
    \begin{align*}
        \sem{t}^{\assign}_{B,f} &= \sem{\op}_{B,f}(\sem{t_1}^{\assign}_{B,f},\dots, \sem{t_n}^{\assign}_{B,f})\\
        &= \sem{\op}_B(\sem{t_1}^{\assign}_{B,f},\dots, \sem{t_n}^{\assign}_{B,f})\\
        &= \sem{\op}_B(\sem{\term{\sig}f(t_1)}_B,\dots, \sem{\term{\sig}f(t_n)}_B)\\
        &= \sem{\term{\sig}f(t)}_B.
    \end{align*}
\end{proof}

\begin{lemma}
    The map $\widehat{\eta}_{(A,d)}$ is nonexpansive.
\end{lemma}
\begin{proof}
    Apply Lemma \ref{lem:nonexpextendedalg} to the term algebra $(\qterm{\qsig,S}B,\qterm{\qsig,S}\Delta, \sem{\qsig})$ and the map $f': A \rightarrow \qterm{\qsig,S}B$ defined by $a \mapsto [f(a)]_{\termeq{B}}$ which is nonexpansive as 
    \[\qterm{\qsig,S}\Delta([f(a)]_{\termeq{B}},[f(a')]_{\termeq{B}})\leq \Delta(f(a),f(a')) \leq d(a,a').\]\todo{Possible space saving.}
    We find that $(\qterm{\qsig,S}B,\qterm{\qsig,S}\Delta, \sem{\qsig}_f)$ satisfies all the clauses in $S_A$ and for any $t \in \term{\sig}A$, $\sem{t}^\assign_f = \sem{\term{\sig}f(t)} = [\term{\sig}f(t)]_{\termeq{B}}$ (with $\assign: \emptyset \rightarrow \qterm{\qsig,S}B$). We obtain the following implications (we leave the equivalences implicit) which prove the lemma.
    \begin{align*}
        &[s] = [t] &&\qterm{\qsig,S}d([s], [t])\leq \varepsilon\\
        &\Leftrightarrow \emptyset \relgen{S_A} s = t &&\Leftrightarrow \emptyset \relgen{S_A} s =_\varepsilon t\\
        &\Rightarrow \sem{s}^\assign_f = \sem{t}^\assign_f &&\Rightarrow \qterm{\qsig,S}\Delta(\sem{s}^\assign_f, \sem{t}^\assign_f)\leq \varepsilon\\
        &\Rightarrow [\term{\sig}f(s)] = [\term{\sig}f(t)] &&\Rightarrow \qterm{\qsig,S}\Delta([\term{\sig}f(s)], [\term{\sig}f(t)])\leq \varepsilon
    \end{align*}
\end{proof}
\begin{lemma}\label{lem:munonexp}
    The map $\widehat{\mu}_{(A,d)}$ is well-defined on equivalence classes and nonexpansive.
\end{lemma}
\begin{proof}
    Apply Lemma \ref{lem:nonexpextendedalg} to the term algebra $(\qterm{\qsig,S}A,\qterm{\qsig,S}d, \sem{\qsig})$ and the identity $\id: \qterm{\qsig,S}A \rightarrow \qterm{\qsig,S}A$ which is nonexpansive. We find that $(\qterm{\qsig,S}A,\qterm{\qsig,S}d, \sem{\qsig}_{\id})$ satisfies all the clauses in $S_{\qterm{\qsig,S}A}$ and for any $t \in \term{\sig}(\qterm{\qsig}A)$, $\sem{t}^\assign_{\id} = \sem{\term{\sig}\id(t)} = [\term{\sig}\id(t)]_{\termeq{A}}$ (with $\assign: \emptyset \rightarrow \qterm{\qsig,S}A$). We obtain the following implications which prove nonexpansiveness of $\widehat{\mu}_{(a,d)}$ (well-definedness is proven similarly).
    \begin{align*}
        &\qterm{\qsig,S}\left( \qterm{\qsig,S}d \right)([s]_{\termeq{\qterm{\qsig,S}A}}, [t]_{\termeq{\qterm{\qsig,S}A}})\leq \varepsilon\\
        &\Leftrightarrow \emptyset \relgen{S_{\qterm{\qsig,S}A}} s =_\varepsilon t\\
        &\Rightarrow \qterm{\qsig,S}d(\sem{s}^\assign_{\id}, \sem{t}^\assign_{\id})\leq \varepsilon\\
        &\Rightarrow \qterm{\qsig,S}d([\term{\sig}\id(s)]_{\termeq{A}}, [\term{\sig}\id(t)]_{\termeq{A}})\leq \varepsilon\\
        &\Rightarrow \qterm{\qsig,S}d(\widehat{\mu}_{(A,d)}([s]_{\termeq{\qterm{\qsig,S}A}}), \widehat{\mu}_{(A,d)}([t]_{\termeq{\qterm{\qsig,S}A}}))\leq \varepsilon
    \end{align*}
    The last implication holds by $[\term{\sig}\id(t)]_{\termeq{A}} = \widehat{\mu}_{(A,d)}([t]_{\termeq{\qterm{\qsig,S}A}})$.
\end{proof}
The associativity of $\widehat{\mu}$, the unitality of $\widehat{\eta}$ and the naturality of both all follow from their counterpart for $\mu^{\sig}$ and $\eta^{\sig}$. This concludes the proof of Theorem \ref{thm:qtermmonad}.

\subsubsection{Proof of Theorem \ref{EM:theorem}}
Let $E=\mathrm{E}(\vdash_S)$. In order to construct the isomorphism $\EM(\qterm{\qsig,S}) \cong \Alg(\qsig,S)$ we 
will make use of the isomorphism $P:\EM(\term{\sig,E}) \cong \Alg(\sig,E):P^{-1}$ which exists by Proposition \ref{prop:algtermisalgsigeq}. We will also make use of the forgetful functor $\forget : \Alg(\qsig,S) \rightarrow \Alg(\sig,E)$ 
    that forgets about the generalized metric space structure and the fact some clauses in $S$ are satisfied. Our proof that $\EM(\qterm{\qsig,S}) \cong \Alg(\qsig,S)$ is divided in three key steps. 
    
In (Step 1) we construct a functor $F: \EM(\qterm{\qsig,S}) \rightarrow \EM(\term{\sig,E})$ so that we have the following picture.

\begin{equation}\label{diag:summary}
\begin{tikzcd}
    {\EM(\term{\qsig,S})} & {\Alg(\qsig,S)} \\
    {\EM(\term{\sig,E})} & {\Alg(\sig,E)}
    \arrow["F"', from=1-1, to=2-1]
    \arrow["U", from=1-2, to=2-2]
    \arrow["P"', shift right=1, from=2-1, to=2-2]
    \arrow["{P^{-1}}"', shift right=1, from=2-2, to=2-1]
\end{tikzcd}
\end{equation}
In (Step 2), we prove that:
\begin{enumerate}
\item for any $(A,d,\alpha) \in \EM(\qterm{\qsig,S})$, there exists $(A,d,\sem{\qsig}_{\alpha}) \in \Alg(\qsig,S)$ such that $PF(A,d,\alpha) = \forget(A,d,\sem{\qsig}_{\alpha})$, and 
\item for any $\algA = (A,d,\sem{\qsig}) \in \Alg(\qsig,S)$, there exists $(A,d,\alpha_{\algA})$ such that $P^{-1}\forget(\algA) = F(A,d,\widehat{\alpha}_\algA)$. 
\end{enumerate}
Finally, in (Step 3) we conclude $\widehat{P}$ and $\widehat{P}^{-1}$ acting trivially on morphisms and as below on objects (with the notation introduced in Step 2) define functors that are inverse to each other.
\begin{align*}
    \widehat{P}:&\EM(\qterm{\qsig,S}) \rightarrow \Alg(\qsig,S)&\EM(\qterm{\qsig,S}) \leftarrow \Alg(\qsig,S):\widehat{P}^{-1}\\
    &(A,d,\alpha) \mapsto (A,d,\sem{\qsig}_{\alpha})&(A,d,\widehat{\alpha}_{\algA}) \mapsfrom (A,d,\sem{\qsig})= \algA\\
\end{align*}
Unrolling the definitions, it will emerge that \eqref{diag:liftingpres} commutes.
\begin{equation}\label{diag:liftingpres}
    \begin{tikzcd}
        {\EM(\term{\qsig,S})} & {\Alg(\qsig,S)} \\
        {\EM(\term{\sig,E})} & {\Alg(\sig,E)}
        \arrow["F"', from=1-1, to=2-1]
        \arrow["U", from=1-2, to=2-2]
        \arrow["P"', shift right=1, from=2-1, to=2-2]
        \arrow["{P^{-1}}"', shift right=1, from=2-2, to=2-1]
        \arrow["{\widehat{P}}"', shift right=1, from=1-1, to=1-2]
        \arrow["{\widehat{P}^{-1}}"', shift right=1, from=1-2, to=1-1]
    \end{tikzcd}
\end{equation}
Then, using the fact that $P$ and $P^{-1}$ are inverses and the fact that the vertical functors forget information that is not modified by $\widehat{P}$ and $\widehat{P}^{-1}$, we can infer the latter pair are inverses too. We conclude $\EM(\qterm{\qsig,S}) \cong \Alg(\qsig,S)$.

\paragraph{Step 1: Construction of $F$.} The functor $F$ will be constructed, by application of Proposition \ref{prop:monadfunctorfunctor} ($F= (\forget-\circ \collapse{})$), by proving that $(U,\collapse{})$ is a monad functor from $\qterm{\qsig,S}$ to $\term{\sig,E}$ where $U:\GMet\rightarrow \Set$ is the forgetful functor and the natural transformation $$\collapse{(A,d)}: \term{\sig,E}A \rightarrow \term{\sig}A/{\termeq{S_A}}$$ is defined as:
$$
[t]_E   \stackrel{\collapse{(A,d)}}{\mapsto}[t]_{\termeq{S_A}},
$$
where $\termeq{S_A}=\mathrm{E}(\vdash_{S_A})$ is the set of equations in the theory $\vdash_S$ extended by $(A,d)$.\footnote{It is a non-trivial observation that $E$ can be a strict subset of $\mathrm{E}(\vdash_{S_A})$. An instance of this happening is in the theory of convex semilattices with black-hole (see Theorem 44 of \cite{DBLP:conf/lics/MioSV21}).} To lighten the notation, we write $[t]$ as a shorthand for $[t]_{E}$, and $\qeqclass{t}$ for either $[t]_{\termeq{S_A}}$ or $[t]_{\termeq{S_{ \qterm{\qsig,S}(A,d)}}}$ with the context making it clear which of the two is intended. Therefore the action of $\collapse{}$ always looks like \[[t] \stackrel{\collapse{(A,d)}}{\mapsto} \qeqclass{t}.\]

We will show that $(U,\collapse{})$ is a monad functor from $\qterm{\qsig,S}$ to $\term{\sig,E}$. First, we show $\collapse{}: \term{\sig,E}\forget \Rightarrow \forget\qterm{\qsig,S}$ is natural, i.e.: it makes \eqref{diag:naturalp} commute for any $f:(A,d) \rightarrow (B,\Delta)$.
\begin{equation}\label{diag:naturalp}
    \begin{tikzcd}
        {\term{\sig,E}A} & {\term{\sig,E}B} \\
        {\term{\sig}A/{\termeq{S_A}}} & {\term{\sig}B/{\termeq{S_B}}}
        \arrow["{\term{\sig,E}f}", from=1-1, to=1-2]
        \arrow["{\collapse{(A,d)}}"', from=1-1, to=2-1]
        \arrow["{\qterm{\qsig,S}f}"', from=2-1, to=2-2]
        \arrow["{\collapse{(B,\Delta)}}", from=1-2, to=2-2]
    \end{tikzcd}
\end{equation}
Starting with $[t]$ in the top left, the bottom path yields $\qeqclass{t}$ then $\qterm{\qsig,S}f(\qeqclass{t}) = \qeqclass{\term{\sig}f(t)} $  and the top path yields $\term{\sig,E}f([t]) = [\term{\sig}f(t)]$ then $\qeqclass{\term{\sig}f(t)}$. Second, we show that $\collapse{} \cdot \eta^{\sig,E} \forget = \forget\widehat{\eta}$. At component $(A,d)$ and for any $a \in A$, we have
\[\collapse{(A,d)}(\eta_A(a)) = \collapse{(A,d)}([a]) =\qeqclass{a} = \forget\widehat{\eta}_{(A,d)}(a).\]
Finally, we prove that \eqref{diag:halfliftp} commutes as follows: 
\begin{equation}\label{diag:halfliftp}
{\fontsize{8.4}{8.4}
    \begin{tikzcd}
        {\term{\sig,E}\term{\sig,E}A} && {\term{\sig,E}A} \\
        {\term{\sig,E}(\term{\sig}A/{\termeq{S_A}})} & {\term{\sig}(\term{\sig}A/{\termeq{S_A}})/{\termeq{S_{\qterm{\qsig,S}(A,d)}}}} & {\term{\sig}A/{\termeq{S_A}}}
        \arrow["{\collapse{(A,d)}}", from=1-3, to=2-3]
        \arrow["{\collapse{\qterm{\qsig,S}(A,d)}}"', from=2-1, to=2-2]
        \arrow["{\widehat{\mu}_{(A,d)}}"', from=2-2, to=2-3]
        \arrow["{\mu^{\sig,E}_A}", from=1-1, to=1-3]
        \arrow["{\term{\sig,E}\collapse{(A,d)}}"', from=1-1, to=2-1]
    \end{tikzcd}
    }
\end{equation}
\begin{equation*}
{\fontsize{8.4}{8.4}
\begin{tikzcd}
    {[t([t_1],\dots,[t_n])]} && {[t(t_1,\dots,t_n)]} \\
    {[t(\qeqclass{t_1},\dots,\qeqclass{t_n})]} & {\qeqclass{t(\qeqclass{t_1},\dots,\qeqclass{t_n})}} & {\qeqclass{t(t_1,\dots,t_n)}}
    \arrow[maps to, from=1-1, to=2-1]
    \arrow[maps to, from=2-1, to=2-2]
    \arrow[maps to, from=2-2, to=2-3]
    \arrow[maps to, from=1-1, to=1-3]
    \arrow[maps to, from=1-3, to=2-3]
\end{tikzcd}
}
\end{equation*}

Concretely, the functor $F= (\forget-\circ \collapse{})$, which is indeed a functor by Proposition \ref{prop:monadfunctorfunctor}, acts as follows on objects:
\begin{center}
\begin{tabular}{c c c}
$\EM(\qterm{\qsig,S})$ &  & $ \EM(\term{\sig,E})$\\
$\alpha: \qterm{\qsig,S}(A,d) \rightarrow (A,d)  $ & $\mapsto$ & $ U\alpha \circ \collapse{(A,d)}: \term{\sig,E}A \rightarrow A$
\end{tabular}
\end{center}
It acts trivially on morphisms, namely if $f:A \rightarrow B$ is a $\qterm{\qsig,S}$--algebra homomorphism $(A,d,\alpha) \rightarrow (B,\Delta,\beta)$, then it is sent to $f:A \rightarrow B$ which is a $\term{\sig,E}$--algebra homomorphism $F(A,d,\alpha) \rightarrow F(B,\Delta,\beta)$.

\paragraph{Step 2.1: the functor $\widehat{P}:\EM(\qterm{\qsig,S}) \rightarrow \Alg(\qsig,S)$.}

Let $\alpha: \qterm{\qsig,S}(A,d) \rightarrow (A,d)$ be in $\EM(\qterm{\qsig,S})$ and denote $(A,\sem{\qsig}_{\alpha})$ the $\sig$--algebra obtained from applying $P$ to $\forget\alpha \circ \collapse{(A,d)}$. Explicitly, for each $\op:n \in \sig$, $\sem{\op}_{\alpha}$ sends $(a_1,\dots,a_n)$ to $\alpha(\qeqclass{\op(a_1,\dots,a_n)})$. We claim that $\algA_\alpha := (A,d,\sem{\qsig}_{\alpha}) \in \Alg(\qsig,S)$. Namely, for any $\op:n:L \in \qsig$, $\sem{\op}_{\alpha}$ is nonexpansive with respect to the lifting $L(A,d)$ and $\algA_\alpha \satisfies S$.

First, we show $\sem{\op}_{\alpha}$ is nonexpansive. Given $\vec{a}, \vec{b} \in L(A,d)$, let $\Delta$ be the restriction of $d$ on $\vec{a}\cup \vec{b}$, we have
\begin{align*}
    d(\alpha(\qeqclass{\op(\vec{a})}), \alpha(\qeqclass{\op(\vec{b})})) &\leq \qterm{\qsig,S}d(\qeqclass{\op(\vec{a})},\qeqclass{\op(\vec{b})})\\
    &\leq L(\Delta)(\vec{a},\vec{b})\\
    &= L(d)(\vec{a},\vec{b}).
\end{align*}
The first inequality holds because $\alpha$ is nonexpansive, the second inequality uses the rule $L$--\textsf{NE}, and the equality is the fact that $L$ preserves isometric embeddings. 

Next, we show $\algA_\alpha$ satisfies $S$. We can view any assignment $\assign: X \rightarrow A$ as an assignment $\assign: X \rightarrow \qterm{\qsig,S}A$, and to distinguish the two extensions to arbitrary $\sig$--terms, we write
\[\sem{-}_{\alpha}^{\assign}: \term{\sig}X \rightarrow A\text{ and }
\assign^*:\term{\sig}X \rightarrow \qterm{\qsig,S}A.\]
We claim that for any basic (quantitative) equation $\eqn \in \eqns{\sig}X$, \[\algA_{\alpha} \satp{\assign} \eqn  \implies  \qterm{\qsig,S}(A,d) \satp{\assign} \eqn.\]
Indeed, for any $x \in X$, we have $\sem{x}_{\alpha}^{\assign} = \assign(x) = \assign^*(x)$, thus if $\eqn$ is quantitative, w.l.o.g. it is $x=_{\varepsilon}y$, the following implications hold:
\begin{align*}
    \algA_{\alpha} \satp{\assign} \eqn &\Leftrightarrow d(\sem{x}_{\alpha}^{\assign}),\sem{y}_{\alpha}^{\assign}))\leq \varepsilon\\
    &\Leftrightarrow d(\assign(x),\assign(y))\leq \varepsilon\\
    &\Rightarrow \emptyset \vdash_{S_A} \assign(x) =_{\varepsilon} \assign(y)\\
    &\Leftrightarrow \qterm{\qsig,S}d(\assign^*(x),\assign^*(y))\leq \varepsilon\\
    &\Leftrightarrow \qterm{\qsig,S}(A,d) \satp{\assign} \eqn.
\end{align*}
The non-invertible implication holds because $\assign(x)$ and $\assign(y)$ are elements of $A$, so the clause $\emptyset \implies \assign(x)=_{d(\assign(x),\assign(y))} \assign(y)$ belongs to $S_A$. Using \textsf{Max} yields $\emptyset \vdash_{S_A} \assign(x) =_{\varepsilon} \assign(y)$. A very similar argument works when $\eqn$ is not quantitative.

Let $\bigwedge_{i\in I}\eqn_i \Rightarrow \eqn$ be a clause in $S$ (each $\eqn_i$ is basic), and suppose $\algA_{\alpha} \satp{\assign}\eqn_i$ for each $i \in I$. By our argument above, we also have $\qterm{\qsig,S}(A,d) \satp{\assign}\eqn_i$, and since $\qterm{\qsig,S}(A,d) \in \Alg(\qsig,S)$, it satisfies all clauses in $S$. We infer $\qterm{\qsig,S}(A,d) \satp{\assign}\eqn$. Now, suppose $\eqn$ has the shape $s=_\varepsilon t$ (a very similar argument will work if $\eqn$ is not quantitative), we have $\qterm{\qsig,S}d(\assign^*(s),\assign^*(t)) \leq \varepsilon$ and since $\alpha$ is nonexpansive, we also have $d(\alpha(\assign^*(s)),\alpha(\assign^*(t))) \leq \varepsilon$. Now, one can show by induction that $\sem{-}^{\assign}_{\alpha} = \alpha \circ \assign^*$, thus $\algA_{\alpha} \satp{\assign} s=_\varepsilon t$. We conclude $\algA \satisfies S$.

This describes the action of $\widehat{P}$ on objects. On morphisms, we said the action is trivial because if $f:(A,d_A,\alpha) \rightarrow (B,d_B,\beta)$ is a homomorphism of $\qterm{\qsig,S}$--algebras, then the underlying function $f:A \rightarrow B$ is nonexpansive and it is a homomorphism of $\sig$--algebras $(A,\sem{\sig}_\alpha) \rightarrow (B,\sem{\sig}_\beta)$. Therefore, it is also a $(\qsig,S)$--algebra homomorphism $(A,d_A,\sem{\sig}_\alpha) \rightarrow (B,d_B,\sem{\sig}_\beta)$. Functoriality is easy to check.

\paragraph{Step 2.2: the functor $\widehat{P}:\EM(\qterm{\qsig,S}) \rightarrow \Alg(\qsig,S)$.}

Let $\algA = (A,d,\sem{\qsig})$ be in $\Alg(\qsig,S)$ and denote $\alpha_\algA:\term{\sig,E}A \rightarrow A$ the $\term{\sig,E}$--algebra obtained from applying $P^{-1}$ to $\forget\algA$. We claim that $\alpha_\algA$ is in the image of $\forget-\circ \collapse{}$. We first show $\alpha_\algA$ is compatible with $\termeq{S_A}$ and nonexpansive with respect to $\termd{S_A}$.

For the former, suppose that $\emptyset \relgen{S_A} s=t$ with $s,t \in \term{\sig}A$. Setting $\sem{a} = a$ for every $a \in A$, we can check that $\algA^+ := (A,d,\sem{\qsig_A}) \in \Alg(\qsig_A,S_A)$. Therefore, by Theorem \ref{thm:soundness}, we have $\algA^+ \satisfies s=t$. Thus, for the only possible assignment $\assign: \emptyset \rightarrow A$ (recall that $s, t \in \term{\sig}A \subseteq \term{\sig_A}\emptyset$), we find
\[\alpha_\algA(s) = \sem{s}^{\assign} = \sem{t}^{\assign} = \alpha_{\algA}(t).\]

For the latter, we can use the same reasoning starting with the assumption $\emptyset \relgen{S_A} s=_\varepsilon t$ to obtain $d(\alpha_\algA(s),\alpha_\algA(t))\leq \varepsilon$.

We now have a nonexpansive map $\widehat{\alpha}_\algA: \qterm{\qsig,S}(A,d) \rightarrow (A,d)$ defined by $\widehat{\alpha}_\algA(\qeqclass{t}) = \alpha_\algA(t)$. Equivalently, it can be inductively defined: for any $a \in A$, $\widehat{\alpha}_\algA(\qeqclass{a}) = a$ and $\forall \op:n \in \sig$, $\forall t_1,\dots, t_n\in \term{\sig}A$,
\[\widehat{\alpha}_\algA(\qeqclass{\op(t_1,\dots,t_n)}) = \sem{\op}\left( \widehat{\alpha}_\algA(\qeqclass{t_1}),\dots, \widehat{\alpha}_\algA(\qeqclass{t_n}\qeqclass) \right).\] It remains to show it is a $\qterm{\qsig,S}$--algebra. This is a direct consequence of $\alpha_\algA$ being a $\term{\sig,E}$--algebra. Indeed, for any $a \in A$, we have
\begin{align*}
    \widehat{\alpha}_\algA(\widehat{\eta}_{(A,d)}(a)) &= \widehat{\alpha}_\algA(\qeqclass{a})\\
    &= \alpha_\algA(a) = a,
\end{align*}
and for any $\qeqclass{t(\qeqclass{t_1},\dots,\qeqclass{t_n})} \in $, we have
\begin{align*}
    &\widehat{\alpha}_\algA(\qterm{\qsig,S}(\widehat{\alpha}_\algA)\left( \qeqclass{t(\qeqclass{t_1},\dots,\qeqclass{t_n})}\right))\\
    &= \widehat{\alpha}_\algA \left( \qeqclass{t(\alpha_\algA(t_1),\dots,\alpha_\algA(t_n))} \right)\\
    &= \alpha_\algA \left(t(\alpha_\algA(t_1),\dots,\alpha_\algA(t_n)) \right)\\
    &= \alpha_\algA(t(t_1,\dots,t_n))\\
    &= \widehat{\alpha}_\algA(\qeqclass{t(t_1,\dots,t_n)})\\
    &= \widehat{\alpha}_\algA(\widehat{\mu}_{(A,d)}\left( \qeqclass{t(\qeqclass{t_1},\dots,\qeqclass{t_n})}\right))\\
\end{align*}

This describes the action of $\widehat{P}^{-1}$ on objects. On morphisms, an argument similar the one above yield the functor $\widehat{P}^{-1}$.

\subsection{Proofs of Section \ref{examples}}
In this Section, whenever $p \in (0,1)$, we denote $\overline{p} = 1-p$.
\subsubsection{Proof of Proposition \ref{prop:LKdiffuse}}

Let $\distr, \distrb \in \mD A$, since $\LK{d}(\distr,\distrb)$ is a sum of products of numbers in $[0,1]$, we find that $\LK{d}$ has type $\mD A\times \mD A \rightarrow [0,1]$, i.e. it is a fuzzy relation. It is also clear that $\LK{d}$'s definition does not depend on the order of the inputs, so it is symmetric \eqref{eq:symm}. For the triangle inequality \eqref{eq:trineq}, we have the following derivation for all $\distr,\distrb,\distrc \in \mD A$, where $x$, $y$, and $z$ range in $\supp{\distr}$, $\supp{\distrb}$ and $\supp{\distrc}$ respectively.
\begin{align*}
    &\LK{d}(\distr, \distrb) + \LK{d}(\distrb, \distrc)\\
    &= \sum_{(x,y)} \distr(x)\distrb(y)d(x,y) + \sum_{(y,z)} \distrb(y)\distrc(z)d(y,z)\\
    &= \sum_y \distrb(y)\left( \sum_x \distr(x)d(x,y) + \sum_z \distrc(z)d(y,z)) \right)\\
    &= \sum_y \distrb(y)\left( \sum_{(x,z)} \distr(x)\distrc(z)d(x,y) + \sum_{(x,z)} \distr(x)\distrc(z)d(y,z)) \right)\\
    &= \sum_y \distrb(y)\left( \sum_{(x,z)} \distr(x)\distrc(z)(d(x,y) + d(y,z)) \right)\\
    &\geq \sum_y \distrb(y)\left( \sum_{(x,z)} \distr(x)\distrc(z)d(x,z)\right)\\
    &= \sum_{(x,z)} \distr(x)\distrc(z)d(x,z) = \LK{d}(\distr,\distrc)
\end{align*}

\subsubsection{Proof of Lemma \ref{lem:LKpliftpresiso}}
Let $(A,d)$ be a diffuse metric space, and $p \in (0,1)$. For any subset $A'\subseteq A$ and any $a,b,a',b' \in A'$,
\begin{align*}
    &L^p_{\textnormal{\L K}}(d)((a,b),(a',b'))\\
    &= \LK{d}(\sem{+_p}(a,b),\sem{+_p}(a',b'))\\
    &=\LK{d}(pa+\overline{p}b,pa'+\overline{p}b')\\
    &= p^2d(a,a')+p\overline{p}d(a,b')+\overline{p}pd(b,a') + \overline{p}^2d(b,b')\\
    &= p^2d|_{A'}(a,a')+p\overline{p}d|_{A'}(a,b')+\overline{p}pd|_{A'}(b,a') + \overline{p}^2d|_{A'}(b,b')\\
    &= \LK{d|_{A'}}(pa+\overline{p}b,pa'+\overline{p}b')\\
    &= \LK{d|_{A'}}(\sem{+_p}(a,b),\sem{+_p}(a',b'))\\  
    &= L^p_{\textnormal{\L K}}(d|_{A'})((a,b),(a',b')).
\end{align*}
In other words, if $i: A' \hookrightarrow A$ is the inclusion function (without loss of generality, these are the only isometric embeddings we need to consider) $L^p_{\textnormal{\L K}}(i)$ is an isometric embedding.

\subsubsection{Proof of Lemma \ref{lem:convcombLKNE}}
For any $\distr, \distr', \distrb,\distrb' \in \mD A$, we have the following derivation where $x$ and $y$ range over the union of the support of all these distributions.
\begin{align*}
    &\LK{d}(\sem{+_p}(\distr,\distrb),\sem{+_p}(\distr',\distrb'))\\
    &= \LK{d}(p\distr+\overline{p}\distrb,p\distr'+\overline{p}\distrb')\\
    &= \sum_{x,y}\left( p\distr(x)+\overline{p}\distrb(x) \right)\left( p\distr'(y)+\overline{p}\distrb'(y) \right)d(x,y)\\
    &= \sum_{x,y}\Big( p^2\distr(x)\distr'(y) + p\overline{p}\distr(x)\distrb'(y) \\
    &\qquad + \overline{p}p\distrb(x)\distr'(y) + \overline{p}^2\distrb(x)\distrb'(y) \Big)d(x,y)\\
    &= p^2\sum_{x,y}\distr(x)\distr'(y)d(x,y) + p\overline{p}\sum_{x,y}\distr(x)\distrb'(y)d(x,y)\\
    &\quad + \overline{p}p\sum_{x,y}\distrb(x)\distr'(y)d(x,y)+ \overline{p}^2\sum_{x,y}\distrb(x)\distrb'(y)d(x,y)\\
    &= p^2\LK{d}(\distr,\distr')+p\overline{p}\LK{d}(\distr,\distrb')\\
    &\quad +\overline{p}p\LK{d}(\distrb,\distr') + \overline{p}^2\LK{d}(\distrb,\distrb')\\
    &= \LK{\LK{d}}(\sem{+_p}(\dirac{\distr},\dirac{\distrb}),\sem{+_p}(\dirac{\distr'},\dirac{\distrb'}))\\
    &= L^p_{\textnormal{\L K}}(\LK{d})\left( (\distr,\distrb),(\distr',\distrb') \right).
\end{align*}

\subsubsection{Proof of Theorem \ref{thm:freealgLK}}
Let $\eta_{(A,d)} : (A,d) \rightarrow (\mD A, \LK{d})$ be defined by $a \mapsto \dirac{a}$, we show that for any $(B,\Delta,\sem{+_p}_B)$ and nonexpansive map $f:(A,d) \rightarrow (B,\Delta)$, there exists a unique homomorphism $f^* : A \rightarrow B$ in $\Alg(\qsig_{\textnormal{\L K}},E)$ such that $f^* \circ \eta = f$. This is summarized in \eqref{diag:freeconvalg}.
\begin{equation}\label{diag:freeconvalg}
\begin{tikzcd}
	{(A,d)} & {(\mD A,\LK{d})} & {(\mD A,\LK{d},\sem{+_p})} & {} \\
	& {(B,\Delta)} & {(B,\Delta,\sem{-}_B)}
	\arrow["{\eta_{(A,d)}}", from=1-1, to=1-2]
	\arrow["f"', from=1-1, to=2-2]
	\arrow[""{name=0, anchor=center, inner sep=0}, "{f^*}", dashed, from=1-3, to=2-3]
	\arrow["{\text{in }\DMet}", shift left=5, draw=none, from=1-1, to=1-2]
	\arrow[""{name=1, anchor=center, inner sep=0}, "{f^*}", dashed, from=1-2, to=2-2]
	\arrow["{\text{in }\Alg(\qsig_{\textnormal{\L K}},E)}", shift left=5, draw=none, from=1-2, to=1-4]
	\arrow["\forget"', shorten <=16pt, shorten >=16pt, from=0, to=1]
\end{tikzcd}
\end{equation}
First, we show that $\eta_{(A,d)}$ is nonexpansive. Let $a, a' \in (A,d)$, the following derivation shows $\eta^{\textnormal{\L K}}_{(A,d)}$ is an isometry.
\begin{align*}
    \LK{d}(\dirac{a}, \dirac{a'}) &= \sum_{x \in \supp{\dirac{a}}}\sum_{y \in \supp{\dirac{a'}}}\dirac{a}(x)\dirac{a'}(y)d(a,a')\\
    &= \dirac{a}(a)\dirac{a'}(a')d(a,a')\\
    &= d(a,a')
\end{align*}
We already know $(\mD A, \sem{+_p})$ is a convex algebra and Lemma \ref{lem:convcombLKNE} tells us each $\sem{+_p}: L^p_{\textnormal{\L K}}(\mD A, \LK{d}) \rightarrow (\mD A, \LK{d})$ is nonexpansive, thus $(\mD A, \LK{d}, \sem{+_p})$ is a convex \L K algebra. Now, since $\mD A$ is the free convex algebra on $A$ and $a \mapsto \dirac{a}$ is the universal morphism witnessing this, we have convex algebra homomorphism $f^*:(\mD A, \sem{+_p}) \rightarrow (B,\sem{+_p}_B)$ making the triangle above commute. It remains to show it is nonexpansive to conclude it is a morphism in $\Alg(\qsig_{\textnormal{\L K}},E)$.

Briefly, $f^*$ sends a probability distribution $\distr$ on $A$ to the interpretation in $B$ of a term in $\term{\sig_{\textnormal{\L K}}}A$ corresponding to $\distr$ where every occurence of $a$ has been replaced by $f(a)$. For instance, if $\supp{\distr} = \{a_1,\dots, a_n\}$, one could write
\[f^*(\distr) = \sem{+_{\distr(a_1)}}_B\left( f(a_1), \sem{+_{\frac{\distr(a_2)}{1-\distr(a_1)}}}_B(f(a_2), \dotsb \right).\]
In particular, we have $f^*(\dirac{a}) = f(a)$ for any $a \in A$. Moreover, since $f^*$ is a homomorphism, for any $\distr, \distr' \in \mD A$ and $p \in (0,1)$, we have 
\[f^*(p\distr +\overline{p}\distr') = f^*(\sem{+_p}(\distr,\distr')) = \sem{+_p}_B(f^*(\distr),f^*(\distr').\]
More details can be inferred from \cite{jacobs:2010}.

We are now ready to show $f^*$ is nonexpansive. We proceed by induction on the size of the support of $\distr, \distrb \in \mD A$. For the base case, we must have $\distr = \dirac{a}$ and $\distrb= \dirac{b}$ for $a,b \in A$, then it is easy to compute
\[\Delta(f^*(\dirac{a}),f^*(\dirac{b})) = \Delta(f(a),f(b)) \leq d(a,b) = \LK{d}(\dirac{a},\dirac{b}).\]
Suppose $f^*$ is nonexpansive on all pairs of distributions $\distr$ and $\distrb$ with $2 < |\supp{\distr}|+|\supp{\distrb}| < n$, and fix any $\distr,\distrb \in \mD A$ with  $|\supp{\distr}|+|\supp{\distrb}| = n$. It is always possible to rewrite $\distr = p\dirac{a}+\overline{p}\distr'$ and $\distrb = p\dirac{b}+\overline{p}\distrb'$ such that $|\supp{\distr'}|+|\supp{\distrb'}| < n$ (without loss of generality, we can pick $a$ that has the smallest weight $p$ in $\distr$ and $b$ has weight at least $p$ in $\distrb$).

By the induction hypothesis, we have the following inequalities (recalling that $f^*(\dirac{a}) = f(a)$ and $f^*(\dirac{b}) = f(b)$).
\begin{align*}
    \Delta(f(a),f(b)) &\leq \LK{d}(\dirac{a},\dirac{b})\\
    \Delta(f(a),f^*(\distrb')) &\leq \LK{d}(\dirac{a},\distrb')\\
    \Delta(f^*(\distr'),f(b)) &\leq \LK{d}(\distr',\dirac{b})\\
    \Delta(f^*(\distr'), f^*(\distrb')) &\leq \LK{d}(\distr',\distrb')
\end{align*}
Then, we have the following derivation where $x$ and $y$ range over $\supp{\distr}$ and $\supp{\distr'}$ respectively.
\begin{align*}
    &\Delta(f^*(\distr),f^*(\distrb))\\
    &= \Delta\left( \sem{+_p}_B(f^*(\dirac{a}),f^*(\distr')),\sem{+_p}_B(f^*(\dirac{b}),f^*(\distrb')) \right)\\
    &\leq L^p_{\textnormal{\L K}}(\Delta)\left( (f^*(\dirac{a}),f^*(\distr')), (f^*(\dirac{b}),f^*(\distrb')) \right)\\
    &= L^p_{\textnormal{\L K}}(\Delta)\left( (f(a),f^*(\distr')), (f(b),f^*(\distrb')) \right)\\
    &= \LK{\Delta}\left( pf(a)+\overline{p}f^*(\distr'), pf(b)+\overline{p}f^*(\distrb') \right)\\
    &= p^2\Delta(f(a),f(b))+p\overline{p}\Delta(f(a),f^*(\distrb'))\\
    &\quad +\overline{p}p\Delta(f^*(\distr'), f(b)) + \overline{p}^2\Delta(f^*(\distr'),f^*(\distrb'))\\
    &\leq p^2\LK{d}(\dirac{a},\dirac{b})+p\overline{p}\LK{d}(\dirac{a},\distrb')\\
    &\quad +\overline{p}p\LK{d}(\distr',\dirac{b}) + \overline{p}^2\LK{d}(\distr',\distrb')\\
    &= p^2\sum_{x,y}\dirac{a}(x)\dirac{b}(y)d(x,y) + p\overline{p}\sum_{x,y}\dirac{a}(x)\distrb'(y)d(x,y)\\
    &\quad + \overline{p}p\sum_{x,y}\distr(x)\dirac{b}(y)d(x,y) + \overline{p}^2\sum_{x,y}\distr(x)\distrb'(y)d(x,y)\\
    &= \sum_{x,y}(p\dirac{a}(x)+\overline{p}\distr'(x))(p\dirac{b}(y)+\overline{p}\distrb'(y))d(x,y)\\
    &= \sum_{x,y}\distr(x)\distrb(y)d(x,y)\\
    &= \LK{d}(\distr,\distrb)
\end{align*}
The first inequality holds by $L^p_{\textnormal{\L K}}$--nonexpansiveness of $\sem{+_p}_B$, the second holds by the induction hypothesis (the four inequalities written above).